\documentclass[11pt,a4paper]{article}

\usepackage[utf8]{inputenc}
\usepackage[T1]{fontenc}
\usepackage{lmodern}
\usepackage{geometry}
\usepackage{amsmath,amssymb,amsthm}
\usepackage{graphicx}
\usepackage{booktabs}
\usepackage{hyperref}
\usepackage{amsthm}
\usepackage{aliascnt}
\usepackage[nameinlink,noabbrev]{cleveref}

\newtheorem{theorem}{Theorem}

\newaliascnt{lemma}{theorem}
\newtheorem{lemma}[lemma]{Lemma}
\aliascntresetthe{lemma}

\newaliascnt{proposition}{theorem}

\aliascntresetthe{proposition}

\newaliascnt{corollary}{theorem}
\newtheorem{corollary}[corollary]{Corollary}
\aliascntresetthe{corollary}

\newaliascnt{claim}{theorem}
\newtheorem{claim}[claim]{Claim}
\aliascntresetthe{claim}

\theoremstyle{definition}

\newaliascnt{definition}{theorem}
\newtheorem{definition}[definition]{Definition}
\aliascntresetthe{definition}

\newaliascnt{remark}{theorem}

\aliascntresetthe{remark}

\crefname{theorem}{theorem}{theorems}
\Crefname{theorem}{Theorem}{Theorems}

\crefname{lemma}{lemma}{lemmas}
\Crefname{lemma}{Lemma}{Lemmas}

\crefname{proposition}{proposition}{propositions}
\Crefname{proposition}{Proposition}{Propositions}

\crefname{corollary}{corollary}{corollaries}
\Crefname{corollary}{Corollary}{Corollaries}

\crefname{claim}{claim}{claims}
\Crefname{claim}{Claim}{Claims}

\crefname{definition}{definition}{definitions}
\Crefname{definition}{Definition}{Definitions}

\crefname{remark}{remark}{remarks}
\Crefname{remark}{Remark}{Remarks}

\newenvironment{claimproof}
  {\begin{proof}[Proof of Claim]}
  {\end{proof}}

\input{preamble}

\title{On the Complexity of Secluded Path Problems\footnote{A preliminary version of this paper appeared in the Proceedings of the 20th International Symposium on Parameterized and Exact Computation (IPEC 2025), Vol. 358 of Leibniz International Proceedings in Informatics (LIPIcs), pp. 4:1-4:16, Schloss Dagstuhl–Leibniz-Zentrum für Informatik,  2025~\cite{DBLP:conf/iwpec/HanakaT25}. This work is partially supported by JSPS KAKENHI Grant Numbers
JP21K17707, 
JP22H00513, 
JP25K03077, 
and JST, CRONOS, Japan Grant Number JPMJCS24K2.}}
\author{
  Tesshu Hanaka\thanks{Kyushu University. \texttt{hanaka@inf.kyushu-u.ac.jp}}
  \and
  Daisuke Tsuru\thanks{Kyushu University. \texttt{tsuru.daisuke.740@s.kyushu-u.ac.jp}}
}
\date{}

\begin{document}
\maketitle

\begin{abstract}
This paper investigates the complexity of finding secluded paths in graphs. We focus on the \textsc{Short Secluded Path} problem and a natural new variant we introduce, \textsc{Shortest Secluded Path}. Formally, given an undirected graph $G=(V, E)$, two vertices $s,t\in V$, and two integers $k,l$, the \textsc{Short Secluded Path} problem asks whether there exists an $s$-$t$ path of length at most $k$ with at most $l$ neighbors.
This problem is known to be computationally hard: it is W[1]-hard when parameterized by the path length $k$ or by cliquewidth, and para-NP-complete when parameterized by the number $l$ of neighbors. The fixed-parameter tractability is known for $k+l$ or treewidth.
In this paper, we expand the parameterized complexity landscape by designing (1) an XP algorithm parameterized by cliquewidth and (2) fixed-parameter algorithms parameterized by neighborhood diversity and twin cover number, respectively. As a byproduct, our results also yield parameterized algorithms for the classic \textsc{$s$-$t$ $k$-Path} problem under the considered parameters.
Furthermore, we introduce the \textsc{Shortest Secluded Path} problem, which seeks a shortest $s$-$t$ path with the minimum number of neighbors. In contrast to the hardness of the original problem, we reveal that this variant is solvable in polynomial time on unweighted graphs. We complete this by showing that for edge-weighted graphs, the problem becomes W[1]-hard yet remains in XP when parameterized by the shortest path distance between $s$ and $t$.
\end{abstract}

\section{Introduction}\label{sec:intro}

Path-finding problems are one of the most fundamental graph problems.
They have been well-studied extensively because of their wide range of practical applications~\cite{Dijkstra,Finding_k_paths,ShortestPath_experimental,bellman:routing}.
However, in many real-world scenarios, simply finding a shortest path in the input graph may not be sufficient. For example, in in the context of secure communication, one may seek a transmission path for sensitive information that minimizes exposure to potential eavesdroppers. Similarly, in a transportation network, one may wish to find a convoy route that avoids potential attackers. In robotics, paths with minimal external interference may be required in order to avoid high sensor noise or potential collisions. 

Motivated by these applications, the \textsc{Short Secluded Path} problem has been studied extensively~\cite{algorithmica/ChechikJPP17,mst/FominGKK17,networks/BevernFT20,ipl/LuckowF20}.
A vertex subset $U\subseteq V$ is called \emph{$l$-secluded} if it has at most $l$ neighbors.
Formally, \textsc{Short Secluded Path} is defined as follows.

\problemdef{\textsc{Short Secluded Path}} {An undirected graph $G=(V,E)$ where $|V|=n$ and $|E|=m$, two vertices $s,t\in V$, and two integers $k,l$.}
{Determine whether $G$ has an $l$-secluded $s$-$t$ path of length at most $k$.}

\Cref{fig:ex:ssp} illustrates an $l$-secluded $s$-$t$ path of length at most $k$ for $l=5$ and $k=5$.\footnote{In this paper, the length of a path is defined as the number of vertices on the path.}
Unfortunately, \textsc{Short Secluded Path} is NP-complete in general, since it is equivalent to \textsc{$s$-$t$ Hamiltonian Path} when we set $k=n$ and $l=0$.
Since \textsc{$s$-$t$ Hamiltonian Path} is NP-hard even on restricted graph classes~\cite{ita/MeloFS23,tcs/HanakaK25}, \textsc{Short Secluded Path} is also NP-hard on planar bipartite graphs of maximum degree 3, strongly chordal split graphs, and chordal bipartite graphs.

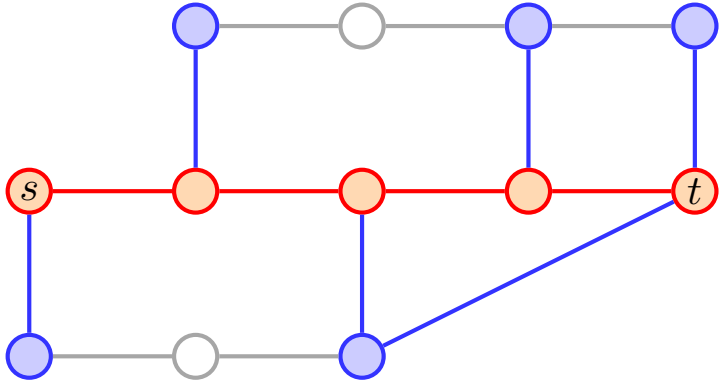
\begin{figure}[t]
    \centering
\resizebox{0.6\linewidth}{!}{%
\begin{tikzpicture}[
    node distance=1cm,
    start_end_node/.style={circle, draw=red, line width=1pt, fill=orange!30, minimum size=3mm, font=\small, text=black},
    path_node/.style={circle, draw=red, line width=1pt, fill=orange!30, minimum size=3mm},
    neighbor_node/.style={circle, draw=blue!80, line width=1pt, fill=blue!20, minimum size=3mm},
    text_node/.style={minimum size=3mm,font=\small, text=black},
    other_node/.style={circle, draw=gray!70, line width=1pt, fill=white, minimum size=3mm},
    path_edge/.style={red, line width=1pt},
    neighbor_edge/.style={blue!80, line width=1pt},
    other_edge/.style={gray!70, line width=1pt}
]

    \node[start_end_node] (s) {};
    \node[text_node] (s_text) {$s$};
    \node[path_node, right=of s] (p1) {};
    \node[path_node, right=of p1] (p2) {};
    \node[path_node, right=of p2] (p3) {};
    \node[start_end_node, right=of p3] (t) {};
    \node[text_node, right=of p3] (t_text) {$t$};

    \draw[path_edge] (s) -- (p1);
    \draw[path_edge] (p1) -- (p2);
    \draw[path_edge] (p2) -- (p3);
    \draw[path_edge] (p3) -- (t);

    \node[neighbor_node, below=of s] (n1) {};
    \node[neighbor_node, above=of p1] (n2) {};
    \node[neighbor_node, below =of p2] (n4) {}; 
    \node[neighbor_node, above=of p3] (n5) {};
    \node[neighbor_node, above=of t] (n6) {};

    \node[other_node, right= of n1] (o1) {};
    \node[other_node, above= of p2] (o2) {}; 

    \draw[neighbor_edge] (s) -- (n1);
    \draw[neighbor_edge] (p1) -- (n2);
    \draw[neighbor_edge] (p2) -- (n4);
    \draw[neighbor_edge] (p3) -- (n5);
    \draw[neighbor_edge] (t) -- (n6);
    \draw[neighbor_edge] (t) -- (n4);

    \draw[other_edge] (n1) -- (o1);
    \draw[other_edge] (n4) -- (o1);
    \draw[other_edge] (n2) -- (o2);
    \draw[other_edge] (n5) -- (o2);
    \draw[other_edge] (n5) -- (n6);
\end{tikzpicture}    
}

    \caption{An illustration of a 5-secluded $s$-$t$ path $P$ of length $5$. The $s$-$t$ path $P$ consists of the red vertices. The blue vertices are neighbors of $P$.}
    \label{fig:ex:ssp}
\end{figure}

Therefore, there has been considerable interest in its parameterized complexity.
For the natural parameters $k$ and $l$, Luckow and Fluschnik \cite{ipl/LuckowF20} show that \textsc{Short Secluded Path} is W[1]-hard when parameterized by $k$, whereas it is fixed-parameter tractable (FPT) when parameterized by $k+l$. Note that the problem is para-NP-complete when parameterized by $l$ due to the hardness of \textsc{$s$-$t$ Hamiltonian Path}.
In \cite{networks/BevernFT20}, van Bevern et al. study the fixed-parameter tractability and kernelization complexity of \textsc{Short Secluded Path} for structural parameters related to tree-like graph structures such as treewidth, vertex cover number, feedback vertex set number, and feedback edge set number.
For treewidth $\tw$, they present an FPT algorithm for \textsc{Short Secluded Path} that runs in $2^{O(\tw)} n^{O(1)}$ time. For kernelization, the problem admits a polynomial kernel when parameterized by the combination of $k, l,$ and feedback vertex set number, and also when parameterized by feedback edge set number alone. In contrast, it does not admit a polynomial kernel when parameterized by vertex cover number, by the combination of $l$ and feedback vertex set number, or by the combination of $k, l,$ and treewidth.

\subsection{Our contribution}
In this paper, we present parameterized algorithms for \textsc{Short Secluded Path} with respect to structural parameters related to dense graph classes such as clique-width, twin cover number, and neighborhood diversity. Our results lead to a more comprehensive understanding of the parameterized complexity of \textsc{Short Secluded Path} under structural graph parameters (see \Cref{fig:graph-parameters}).

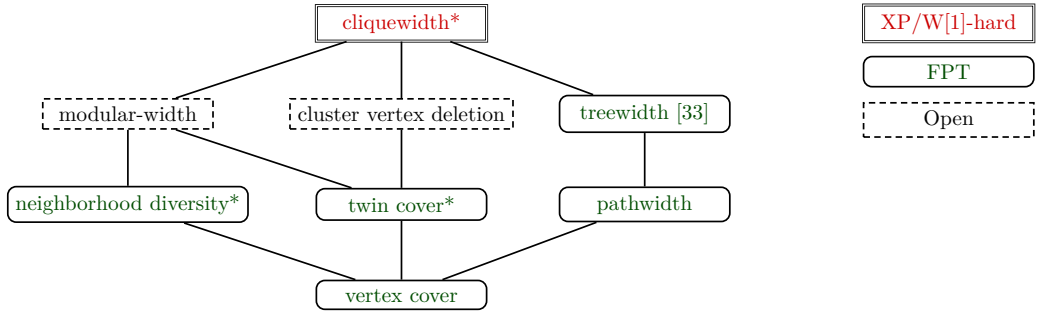
\begin{figure}[ht]
\centering
\scalebox{0.8}{ 
\begin{tikzpicture}[
  every node/.style={draw, rectangle, rounded corners, align=center, minimum width=2.8cm, font=\small},
  XP/.style={text=red!80!black,draw, rectangle,  double, sharp corners},
  FPTpoly/.style={text=blue, draw, rectangle, sharp corners},
  FPT/.style={text=green!30!black,thick},
  Open/.style={text=black,dash pattern=on 3pt off 2pt,sharp corners,thick},
]




\node[XP] (cw) at (-2,4) {cliquewidth*};
\node[FPT] (tw) at (2,2.5) {treewidth \cite{networks/BevernFT20}};
\node[FPT] (pw) at (2,1) {pathwidth};
\node[Open] (mw) at (-6.5,2.5) {modular-width};
\node[Open] (cvd) at (-2,2.5) {cluster vertex deletion};
\node[FPT] (tc) at (-2,1) {twin cover*};
\node[FPT] (nd) at (-6.5,1) {neighborhood diversity*};
\node[FPT] (vc) at (-2,-0.5) {vertex cover};



\draw [thick] (cw) -- (tw);
\draw [thick] (cw) -- (mw);
\draw [thick] (tw) -- (pw);
\draw [thick] (pw) -- (vc);
\draw [thick] (mw) -- (nd);
\draw [thick] (mw) -- (tc);
\draw [thick] (nd) -- (vc);
\draw [thick] (tc) -- (vc);
\draw [thick] (cw) -- (cvd);
\draw [thick] (cvd) -- (tc);


\node[XP] at (7,4) {\textcolor{red!80!black}{XP/W[1]-hard}};
\node[FPT] at (7,3.2) {\textcolor{green!30!black}{FPT}};
\node[Open] at (7,2.4){Open};

\end{tikzpicture}
}
\caption{Parameterized complexity of \textsc{Short Secluded Path} with respect to structural graph parameters. 
The connection between two parameters indicates that the upper parameter \( p \) is bounded by some computable function \( f(\cdot) \) of the lower parameter \( q \); that is, \( p \leq f(q) \). 
Parameters marked with an asterisk (\(*\)) represent our contributions in this paper. 
Double-bordered rectangles indicate that the parameterized problem belongs to XP but is W[1]-hard; 
rounded rectangles indicate that the parameterized problem is fixed-parameter tractable (FPT);
and dotted rectangles represent cases that remain open.
}

\label{fig:graph-parameters}
\end{figure}

First, we propose an XP algorithm when parameterized by clique-width.
Since \textsc{$s$-$t$ Hamiltonian Path} is known to be W[1]-hard when parameterized by clique-width~\cite{siamcomp/FominGLS10,tcs/HanakaK25}, this XP algorithm provides a tight complexity bound for this parameter.
Furthermore, we design a $\nd(G)^{O(\nd(G)^2)}n^{O(1)}$-time algorithm and a $2^{O(\tc(G)^2)}n^{O(1)}$-time algorithm for \textsc{Short Secluded Path} when parameterized by neighborhood diversity $\nd(G)$ and twin cover number $\tc(G)$, respectively. 

Here, it is worth mentioning that our algorithms are designed for
\textsc{Secluded $k$-Path}, which asks whether $G$ has an
$l$-secluded $s$-$t$ path of length \emph{exactly} $k$.
By setting $l=n$, our algorithms also yield parameterized algorithms
for the classical problem \textsc{$s$-$t$ $k$-Path} under the considered parameters.
To the best of our knowledge, the parameterized complexity of
\textsc{$s$-$t$ $k$-Path} with respect to clique-width,
twin cover number, and neighborhood diversity has not been
studied explicitly before.

Then we consider the \textsc{Shortest Secluded Path} problem.
Unlike \textsc{Short Secluded Path}, the goal of this problem is to find a shortest $s$-$t$ path whose seclusion is minimum among all shortest $s$-$t$ paths. Formally, \textsc{Shortest Secluded Path} is defined as follows.

\problemdef{\textsc{Shortest} Secluded Path} {An undirected graph $G=(V,E)$ where $|V|=n$ and $|E|=m$, and two vertices $s,t\in V$.}
{Find a shortest $s$-$t$ path $P$ that minimizes the size of its open neighborhood $|N(V(P))|$ in $G$.}

Interestingly, we show that \textsc{Shortest Secluded Path} is solvable in polynomial time on unweighted graphs, in contrast to the hardness of \textsc{Short Secluded Path}.
However, for positive integer edge-weighted graphs, the problem becomes W[1]-hard when parameterized by the shortest-path distance $d$ between $s$ and $t$. To complement this, we finally present an XP algorithm when parameterized by $d$.

\subsection{Related work}
There has been extensive literature on secluded subgraph problems.
This line of research begins with Chechik et al.~\cite{algorithmica/ChechikJPP17}, who first introduce the notion of \emph{seclusion} for connectivity problems on graphs. Their original formulation, which they call \emph{exposure}, is defined as the size of the closed neighborhood of a vertex subset.
In their seminal work, Chechik et al. consider two problems: \textsc{Secluded Path} and its generalization \textsc{Secluded Steiner Tree}, where the goal is to find a path or a Steiner tree that minimizes its exposure.
The authors show that while \textsc{Secluded Path} is hard to approximate, it is solvable in polynomial time on graphs with bounded degree. For \textsc{Secluded Steiner Tree}, they present a fixed-parameter algorithm parameterized by treewidth.
Subsequently, Fomin et al. show that \textsc{Secluded Steiner Tree} is also FPT when parameterized by the exposure of a solution.

A significant shift in perspective comes from van Bevern et al.~\cite{disopt/BevernFMMSS18}, who decouple the solution size from the exposure, and introduce the notion of \emph{$l$-secludedness}. Within this framework, they investigate the parameterized complexity for finding secluded versions of several fundamental graph structures, such as separators, dominating sets, $\mathcal{F}$-free vertex deletion sets, and independent sets.
Building on this work, Golovach et al.~\cite{jcss/GolovachHLM20} show that \textsc{Connected Secluded $\mathcal{F}$-Free Subgraph} is FPT when parameterized by $l$, and Donkers et al.~\cite{jcss/DonkersJK23} present a faster FPT algorithm for \textsc{Secluded Induced Tree} parameterized by $l$. Recently, the concept of seclusion has been extended further, with Mallek et al.~\cite{corr/abs-2502-06048} studying secluded subgraph problems on directed graphs.

\section{Preliminaries}\label{sec:prelim}

In this paper, we use standard graph-theoretic notations. Let $G=(V, E)$ be an undirected graph, where $V$ is the set of vertices and $E$ is the set of edges. We let $n=|V|$ and $m=|E|$. For a positive integer $n$, let $[n]:=\{1,2,\cdots,n\}$

For a vertex $v \in V$, its neighborhood is $N_G(v) = \{u \in V \mid \{v,u\} \in E\}$. For a vertex subset $U \subseteq V$, $N_G(U) = \{v \in V \setminus U \mid u\in U, \{u,v\} \in E\}$ denotes the set of neighbors of $U$. For simplicity, we sometimes use $N(v)$ and $N(U)$ instead of $N_G(v)$ and $N_G(U)$.
The subgraph induced by $U$ is denoted by $G[U]$. 

The length of a path is defined by the number of vertices on the path.
For two vertices $u,v\in V$, the shortest path distance between $u$ and $v$, denoted by $\dist(u,v)$, is defined as the number of edges (resp., the sum of the weights of edges) in a shortest $u$-$v$ path on unweighted graphs (resp., edge-weighted graphs). 
Two distinct vertices $u$ and $v$ are called \emph{twins} if $N(u)\setminus \{v\}=N(v)\setminus \{u\}$. 
In particular, twins $u$ and $v$ are called \emph{true twins} if the edge $\{u,v\}$ exists, and otherwise called \emph{false twins}.

We assume that the readers are familiar with the
basic notions of parameterized complexity~\cite{sp/CyganFKLMPPS15}. 

\subsection{Cliquewidth}

In this subsection, we define the cliquewidth of a graph $G$. The definition relies on the concept of a $k$-labeled graph~\cite{dam/CourcelleO00}.

\begin{definition}[$k$-labeled graph]
\label{def:k-labeled_graph}
Let $k$ be a positive integer.
A $k$-labeled graph is a pair $(G,\lab_G)$ of a graph $G$ and a function $\lab_G:V\rightarrow \{1,\cdots,k\}$.
\end{definition}

\begin{definition}[Cliquewidth]\label{def:clique_width}
The cliquewidth of a graph $G$, denoted by $\cw(G)$, is the minimum integer $k$ such that a $k$-labeled graph $(G,\lab_G)$ can be constructed by repeatedly applying the following operations.
    \begin{description}
    
        \item[(O1)] Add a new vertex $v$ with label $i\in [k]$.
        
        \item[(O2)] Take the disjoint union $(G\oplus H, \lab_{G\oplus H})$ of two $k$-labeled graphs $(G,\lab_G)$  and $(H,\lab_H)$, with
        \begin{equation*}
        \lab_{G\oplus H}(v)=
                \begin{cases}
                    & \lab_G (v)\ \ \text{if}\ v\in V(G),\\
                    & \lab_H (v)\ \ \text{otherwise.}
                \end{cases}
        \end{equation*}
        
        \item [(O3)] Take distinct labels $i,j\in[k]$ for a $k$-labeled graph $G$, and add an edge between every pair of vertices labeled by $i$ and by $j$.
        
        \item [(O4)] Take distinct labels $i,j\in[k]$ for a $k$-labeled graph $G$, and relabel the vertices of label $i$ to label $j$.
    \end{description}
   
\end{definition}

This construction process for a $k$-labeled graph can be represented by a rooted binary tree, called a \emph{$k$-expression tree}. Each node in this tree corresponds to one of the four operations: an \emph{introduce node} (O1), a \emph{union node} (O2), a \emph{join node} (O3), or a \emph{relabel node} (O4).
The leaves of a $k$-expression tree are always introduce nodes, and conversely, all introduce nodes are leaves. The graph associated with any node is the one constructed by the operations in its subtree, and thus the root of the tree represents the final graph $G$.

A $k$-expression tree is \emph{irredundant} if for each edge $\{u,v\} \in E(G)$, there is exactly one corresponding join node that adds this edge.
Any $k$-expression tree can be transformed into an irredundant one with $O(n)$ nodes in linear time~\cite{dam/CourcelleO00}.
Therefore, we can assume without loss of generality that any given $k$-expression tree is irredundant.

While computing the exact cliquewidth and an optimal expression tree is NP-hard, a polynomial-time approximation algorithm exists. Specifically, a $(2^{\cw(G)+1}-1)$-expression tree for a graph $G$ with cliquewidth $\cw(G)$ can be computed in $O(n^3)$ time~\cite{jct/OumS06,talg/Oum08,siamcomp/HlinenyO08}.

\subsection{Neighborhood diversity and twin cover}

A partition $\mathcal{M} = \{M_1, \ldots, M_r\}$ of $V$ is called a \emph{twin partition} of $G$ if every $M_i$ is a set of twins. We call $M_i$ a \emph{module} of $\mathcal{M}$. Then the neighborhood diversity of $G$ is defined as follows.
\begin{definition}[Neighborhood diversity]\label{def:nd} 
The \emph{neighborhood diversity} $\nd(G)$ of $G$ is the minimum number of modules among all twin partitions of $G$.
\end{definition}

We can compute the  \emph{neighborhood diversity} $\nd(G)$ of $G$ and its twin partition in linear time~\cite{Lampis12,McConnell,Tedder}. By definition, each module forms either a clique (if it contains true twins) or an independent set  (if it contains false twins).

The \emph{quotient graph} corresponding to a twin partition $\mathcal{M} = \{M_1, M_2, \ldots, M_r\}$ is the graph $Q = (\mathcal{M}, E(\mathcal{M}))$, where $E(\mathcal{M}) = \left\{ \{M_i, M_j\} \mid \exists \{u, v\} \in E,\ u \in M_i,\ v \in M_j \right\}$.
We observe that for any two modules $M_i$ and $M_j$, either there are no edges between them, or every vertex in $M_i$ is adjacent to every vertex in $M_j$.

\bigskip
A vertex set $X \subseteq V$ is called a \emph{twin cover} if, for every edge $\{u, v\} \in E$, at least one of the following holds:
(i) $u \in X$ or $v \in X$, or (ii)  $u$ and $v$ are \emph{true twins} (i.e., adjacent and have identical open neighborhoods).
The \emph{twin cover number} of $G$, denoted $\tc(G)$, is the size of a minimum twin cover of $G$.

\subsection{Integer linear programming}
In this subsection, we introduce Integer Linear Programming (ILP) and its fixed-parameter tractability. 

\begin{definition}[$p$-\textsc{Variable Integer Linear Programming Feasibility} ($p$-ILP)]\label{ILP_feasibility}
Given a matrix $A\in \mathbb{Z}^{m\times p}$ and a vector $b\in \mathbb{Z}^{m}$, $p$-\textsc{Variable Integer Linear Programming Feasibility} asks whether there exists a vector $x\in \mathbb{Z}^{p}$ satisfying $Ax\le b$.
\end{definition}
It is known that $p$-ILP is fixed-parameter tractable with respect to the number $p$ of variables.
    
\begin{theorem}[\cite{mor/Lenstra83,combinatorica/FrankT87,mor/Kannan87}]\label{thm:ILP_proposition}
$p$-\textsc{Variable Integer Linear Programming Feasibility} can be solved in $O(p^{2.5p+o(p)}\cdot L)$ time, where $L$ is the number of bits in the input.
\end{theorem}

\section{XP Algorithm Parameterized by Cliquewidth}\label{sec: cw}

In this section, we design an $n^{2^{O(\cw)}}$-time algorithm parameterized by cliquewidth for the problem of determining whether $G$ has an $s$-$t$ path $P$ of length $k$ with exactly $l$ neighbors. By solving this for each $k$ and $l$, \textsc{Secluded $k$-Path} and \textsc{Short Secluded Path} can also be solved in time $n^{2^{O(\cw)}}$.

\begin{theorem}{cwXP}\label{thm:cw_XP}
There is an $n^{2^{O(\cw)}}$-time algorithm that, for a graph $G$ of cliquewidth $\cw$, determines whether it contains an $s$-$t$ path of length $k$ with $l$ neighbors.
\end{theorem}

Our algorithm follows the standard dynamic programming framework on
clique-width expression trees used for \textsc{Hamiltonian Path} by
Espelage et al.~\cite{wg/EspelageGW01}. We extend this framework
to solve \textsc{Secluded $k$-Path}.
It is known that for a graph $G$ with cliquewidth $\cw(G)$, a $(2^{\cw(G)+1}-1)$-expression tree $\mathcal{T}$ with $O(n)$ nodes can be computed in $O(n^3)$ time~\cite{jct/OumS06,talg/Oum08,siamcomp/HlinenyO08}.
For our algorithm, we must distinguish the source $s$ and target $t$. We achieve this by assigning them two new, unique labels. We can thus assume that we are given an $r$-expression tree $\mathcal{T}$ with $r \le 2^{\cw(G)+1}+1$ where $s$ is assigned label $r-1$ and $t$ is assigned label $r$. Moreover, we may assume that these two labels are never used in relabel operations. 

For a given expression tree $\mathcal{T}$, we execute dynamic programming in a bottom-up manner, from the leaves to the root.
For each node $\mu$ in $\mathcal{T}$, let $G_\mu=(V_\mu,E_\mu)$ be the labeled subgraph associated with $\mu$.
The DP table at each node $\mu$ stores boolean values for states defined by a tuple of parameters.
A state corresponds to the properties of a set $\mathcal{P}$ of vertex-disjoint paths within $G_\mu$, denoted by $\mathcal{P}_\mu$.
Intuitively, each path in $\mathcal{P}_\mu$ represents a sub-path of a potential solution in $G_{\mu}$, i.e., an $s$-$t$ path of length $k$ with $l$ neighbors.
Thus, the DP state needs to track the number of vertices of the solution, the number of their neighborhoods, the number of other vertices, and the number of sub-paths with respect to the labels of the endpoints. The information on endpoint labels of sub-paths is used when merging sub-paths at a join node.
We use the notation $V(\mathcal{P}_\mu)$ as a shorthand for $\bigcup_{P\in \mathcal{P}_\mu}V(P)$, the set of all vertices in any path in $\mathcal{P}_\mu$.
The parameters in our DP state are defined as follows:

\begin{itemize}
    \item For each label $i \in \{1, \dots, r\}$:
    \begin{itemize}
        \item $x_i^\mu$: the number of vertices in $V(\mathcal{P}_\mu)$ with label $i$.
        \item $y_i^\mu$: the number of vertices in $N(V(\mathcal{P}_\mu)) \cap V_\mu$ with label $i$.
        \item $z_i^\mu$: the number of vertices with label $i$ in $V_\mu$ that belong to neither $V(\mathcal{P}_\mu)$ nor its neighborhood.
    \end{itemize}
    
    \item For each pair of labels $i,j$ with $1\le i \le j\le r$:
    \begin{itemize}
        \item $p_{ij}^\mu$: the number of vertex-disjoint paths in $\mathcal{P}_\mu$ having one endpoint with label $i$ and the other with label $j$.
    \end{itemize}
\end{itemize}

For convenience, we also use the notation $p_{ij}^\mu$ when $i>j$; in this case,
it refers to the entry $p_{ji}^\mu$ in the DP table.
We call $(x_1^\mu,\dots,x_r^\mu)$ the solution count vector, $(y_1^\mu,\dots,y_r^\mu)$ the neighborhood count vector, $(z_1^\mu,\dots,z_r^\mu)$ the remaining vertex count vector, and  $(p_{11}^\mu,\dots,p_{rr}^\mu)$ the path count vector, respectively.
The DP entry $\mathrm{DP}_\mu[(x_1^\mu,\dots,x_r^\mu), (y_1^\mu,\dots,y_r^\mu), (z_1^\mu,\dots,z_r^\mu), (p_{11}^\mu,\dots,p_{rr}^\mu)]$ is a boolean value. It is \texttt{true} if and only if there exists a set of vertex-disjoint paths $\mathcal{P}_\mu$ in $G_\mu$ that satisfies all of the following conditions:
\begin{enumerate}
    \item For each label $i$, the number of vertices in $V(\mathcal{P}_\mu)$ with label $i$ is exactly $x_i^\mu$, i.e., $|V(\mathcal{P}_\mu)|=\sum^{r}_{i=1}x_i^\mu$.
    \item For each label $i$, the number of vertices in the neighborhood $N(V(\mathcal{P}_\mu)) \cap V_\mu$ with label $i$ is exactly $y_i^\mu$, i.e., $|N(V(\mathcal{P}_\mu))|=\sum^{r}_{i=1}y_i^\mu$.
    \item For each label $i$, the number of vertices with label $i$ in $V_\mu$ that belong to neither $V(\mathcal{P}_\mu)$ nor its neighborhood is exactly $z_i^\mu$, i.e., $|V_\mu\setminus(V(\mathcal{P}_\mu)\cup N(V(\mathcal{P}_\mu)))|=\sum^{r}_{i=1}z_i^\mu$.
    \item For each pair of labels $(i,j)$, the number of paths in $\mathcal{P}_\mu$ with endpoints labeled $i$ and $j$ is exactly $p_{ij}^\mu$.
\end{enumerate}

For convenience, if a set  $\mathcal{P}_\mu$ of paths is empty, $|V(\mathcal{P}_\mu)|=0$. In this case, its neighborhood is also empty, and this corresponds to the state where $x_i^\mu=y_i^\mu=p_{ij}^\mu=0$ for all $i,j$.

From these definitions, $0\le x_i^\mu,y_i^\mu,z_i^\mu,p_{ij}^\mu\le n$ for all parameters.
Thus, the size of the DP table at each node is bounded by $n^{O(r^2)}$.
As initialization, we set all entries of the DP tables to \texttt{false}.

The existence of a solution in the original graph $G$ is determined by checking the DP table of the root node, $\gamma$. 
A key observation is that an $s$-$t$ path of length $k$ with $l$ neighbors exists if and only if there exists a path $P$ whose endpoints are labeled by $r-1$ and $r$ in $G_\gamma$ satisfying $|V(P)|=k,|N(V(P))|=l$, and $|V_\gamma \setminus(V(P)\cup N(V(P)))|=n-k-l$. Note that $s$ and $t$ have the unique labels $r-1$ and $r$, respectively.
Therefore, an $s$-$t$ path of length $k$ with $l$ neighbors exists if and only 
if $\mathrm{DP}_{\gamma}[(x_1^{\gamma},y_1^{\gamma},z_1^{\gamma}),\cdots,(x_r^{\gamma},y_r^{\gamma},z_r^{\gamma}),(p_{11}^{\gamma},\cdots,p^{\gamma}_{rr})]=\texttt{true}$ at the root node $\gamma$ such that:
        \begin{itemize}
            \item $x_{r-1}^\gamma=x_r^\gamma=1$,
            \item $y_{r-1}^\gamma=y_r^\gamma=0$,
            \item $z_{r-1}^\gamma=z_r^\gamma=0$,
            \item $\sum_{i=1}^rx_i^\gamma=k$,
            \item $\sum_{i=1}^ry_i^\gamma=l$,
            \item $\sum_{i=1}^rz_i^\gamma=n-k-l$,
            \item $p_{(r-1)r}^\gamma=1$, and
            \item $p_{ij}^\gamma=0$ for $i,j$ except $i=r-1$ and $j=r$.
        \end{itemize}

In the following, we describe the recursive formulas for updating the DP table at each node of the expression tree.
Intuitively, our DP state tracks, for each label, the number of vertices in the solution, their neighborhoods, and the other vertices, as well as the number of vertex-disjoint paths distinguished by their endpoint labels.

For an introduce node, which adds a single labeled vertex, we only need to generate the initial states corresponding to this new vertex, forming a trivial path of length one, or being an isolated vertex not part of the solution.
For a union node, which combines two disjoint subgraphs, the new DP table is computed by merging states from its two children. As no edges are added between the subgraphs, paths from each part remain separate. Thus, for any pair of valid states (one from each child), we can compute the resulting state's parameters by simple addition.
The update for a relabel node that changes label $i$ to $j$ is also based on addition: for each state, the counts associated with label $i$ are consolidated into the counts for label $j$.
A join node adds edges between all vertices of label $i$ and all vertices of label $j$. These new edges may connect multiple existing paths into longer ones. Although the number of ways to connect paths can be large, this process can be handled in time $n^{O(r^2)}$ by adapting the update rule in the dynamic programming for \textsc{Hamiltonian Path} by Espelage et al.~\cite{wg/EspelageGW01}.
The running time for a join node is determined by iterating through all possible resulting states. Since the size of the DP table is $n^{O(r^2)}$, and we must process each state from the child's table, the update at a join node can be performed in time $n^{O(r^2)}$.
As the expression tree has $O(n)$ nodes, the total running time is dominated by the join operations, yielding $n^{O(r^2)} = n^{2^{O(\cw(G))}}$.

We describe the detailed update at each node. For simplicity, we use $*_\mu$ as a shorthand for the state tuple $[(x_1^\mu,y_1^\mu,z_1^\mu),\cdots,(x_r^\mu,y_r^\mu,z_r^\mu),(p_{11}^\mu,\cdots,p_{rr}^\mu)]$
and write $\mathrm{DP}_{\mu}[*_\mu]$ instead of $\mathrm{DP}_\mu[(x_1^\mu,y_1^\mu,z_1^\mu),\cdots,(x_r^\mu,y_r^\mu,z_r^\mu),(p_{11}^\mu,\cdots,p_{rr}^\mu)]$.

\paragraph*{Introduce node}
In an introduce node $\mu$, we suppose that a vertex $v$ with label $\alpha$ is introduced.
Thus, $G_\mu$ is a graph consisting of only $v$ with label $\alpha$.
Note that if $v$ is either $s$ or $t$, $v$ must be included in $\mathcal{P}$.
Then, $\mathrm{DP}_\mu[*_\mu]$ is \texttt{true} for the following cases.

\medskip 

\noindent \textbf{[Case: $v\in V\setminus\{s,t\}$]} 
    \begin{itemize}
    \item If a vertex $v$ is not included in $\mathcal{P}_\mu$, then $\mathcal{P}_\mu=\emptyset$. In this case $|V(\mathcal{P}_\mu)|=0$, $|N(V(\mathcal{P}_\mu))|=0$, and  $|V_\mu \setminus (V(\mathcal{P}_\mu)\cup N(V(\mathcal{P}_\mu)))|=1$.
    Thus, $\mathrm{DP}_\mu[*_\mu]=\texttt{true}$ if 
    \begin{itemize}
        \item $x_i^\mu=y_i^\mu=0$ for $i\in [r]$,
        \item $z_i^\mu=0$ for $i\in [r]\setminus \{\alpha\}$,
        \item $z_\alpha^\mu=1$,
        \item $p_{ij}^\mu=0$ for $i,j\in [r]$.
    \end{itemize}

    \item If a vertex $v$ is included in $\mathcal{P}_\mu$, $\mathcal{P}_\mu$ consists of  a path $P=\langle v\rangle$ whose endpoints are labeled by $\alpha$ in $G_\mu$, and it satisfies $|V(\mathcal{P}_\mu)|=1,|N(V(\mathcal{P}_\mu))|=0,$ and $|V_\mu \setminus(V(\mathcal{P}_\mu)\cup N(V(\mathcal{P}_\mu)))|=0$. Thus, $\mathrm{DP}_\mu[*_\mu]=\texttt{true}$ if
    \begin{itemize}
        \item  $x_i^\mu=0$ for $i\in [r]\setminus \{\alpha\}$,
        \item $x_\alpha^\mu=1$,
        \item $y_i^\mu=0,z_i^\mu=0$ for $i\in [r]$,
        \item $p_{\alpha \alpha}^\mu=1$,
        \item $p_{ij}^\mu=0$ for $i,j$ except $i=j=\alpha$.
    \end{itemize}
    \end{itemize}

\noindent  \textbf{[Case: $v\in \{s,t\}$]} 
    \begin{itemize}
        \item 
        In this case, $\mathcal{P}_\mu$ must contain $v$. Thus, $\mathrm{DP}_\mu[*_\mu]=\texttt{true}$ if
    \begin{itemize}
        \item  $x_i^\mu=0$ for $i\in [r]\setminus \{\alpha\}$,
        \item $x_\alpha^\mu=1$,
        \item $y_i^\mu=0,z_i^\mu=0$ for $i\in [r]$,
        \item $p_{\alpha \alpha}^\mu=1$,
        \item $p_{ij}^\mu=0$ for $i,j$ except $i=j=\alpha$.
    \end{itemize}
    \end{itemize}

\paragraph*{Union node}
In a union node $\mu$, let $\nu_1$ and $\nu_2$ be its child nodes.
Since $G_\mu$ is the disjoint union of $G_{\nu_1}$ and $G_{\nu_2}$, there is no edge between $V_{\nu_1}$ and $V_{\nu_2}$ in $G_\mu$.
Thus, 
$\mathrm{DP}_\mu[*_\mu]=\texttt{true}$ if and only if there exist the state tuples $*_{\nu_1}$ and $*_{\nu_2}$ such that $DP_{\nu_1}[*_{\nu_1}]=\texttt{true}$, $DP_{\nu_2}[*_{\nu_2}]=\texttt{true}$, and
\begin{itemize}
     \item  $x_i^\mu=x_i^{\nu_1}+x_i^{\nu_2}$ for $i\in [r]$,
     \item $y_i^\mu=y_i^{\nu_1}+y_i^{\nu_2}$ for $i\in [r]$,
     \item $z_i^\mu=z_i^{\nu_1}+z_i^{\nu_2}$ for $i\in [r]$,
     \item  $p_{ij}^\mu=p_{ij}^{\nu_1}+p_{ij}^{\nu_2}$ for $i,j\in [r]$.
\end{itemize}

\paragraph*{Relabel node}
In relabel node $\mu$, we relabel the vertices with label $\alpha$ to label $\beta$. 
Let $\nu$ be its child node.
Since only labels $\alpha$ and $\beta$ of $G_\nu$ are changed, $\mathrm{DP}_\mu[*_\mu]=\texttt{true}$ if and only if there exist the state tuple $*_{\nu}$  such that $DP_{\nu}[*_\nu]=\texttt{true}$ and

\begin{itemize}
    \item $x_i^\mu=x_i^\nu$ for $i\in [r]\setminus\{\alpha,\beta\}$,
    \item $x_\alpha^\mu=0$,
    \item $x_\beta^\mu=x_\beta^{\nu}+x_\alpha^{\nu}$,

    \item $y_i^\mu=y_i^\nu$ for $i\in [r]\setminus\{\alpha,\beta\}$,
    \item $y_\alpha^\mu=0$,
    \item $y_\beta^\mu=y_\beta^{\nu}+y_\alpha^{\nu}$,

    \item $z_i^\mu=z_i^\nu$ for $i\in [r]\setminus\{\alpha,\beta\}$,
    \item $z_\alpha^\mu=0$,
    \item $z_\beta^\mu=z_\beta^{\nu}+z_\alpha^{\nu}$,

    \item $p_{ij}^\mu=p_{ij}^\nu$ for $i,j\in [r]\setminus\{\alpha,\beta\}$,

    \item $p_{\alpha j}^\mu=0$ for all $j\in [r]$,

    \item $p_{\beta\beta}^\mu
    =p_{\beta\beta}^\nu+p_{\alpha\beta}^\nu+p_{\alpha\alpha}^\nu$,

    \item $p_{\beta j}^\mu
    =p_{\beta j}^\nu+p_{\alpha j}^\nu$
    for $j\in [r]\setminus\{\alpha,\beta\}$.
\end{itemize}

\paragraph*{Join node}
A join node $\mu$ with child $\nu$ corresponds to the operation of adding edges between all vertices with label $\alpha$ and all vertices with label $\beta$ to the graph $G_\nu$ to form $G_\mu$.
At this node, paths from the set $\mathcal{P}_\nu$ (a valid path configuration in $G_\nu$) can be merged into longer paths using the newly introduced edges.

Since the join operation introduces edges but does not add or remove vertices, the set of vertices on the paths does not change. Thus, $|V(\mathcal{P}_\mu)|$ is equal to $|V(\mathcal{P}_\nu)|$, which implies $x_i^\mu = x_i^\nu$ for all $i \in \{1, \dots, r\}$. Furthermore, since the new edges only involve labels $\alpha$ and $\beta$, the neighborhoods of vertices with other labels remain unchanged. This means $y_i^\mu = y_i^\nu$ and $z_i^\mu = z_i^\nu$ for all $i \in \{1, \dots, r\} \setminus \{\alpha,\beta\}$.

The update logic for the remaining parameters depends on which labels are present in the path set $\mathcal{P}_\nu$. 

\begin{itemize}
    \item \textbf{Case 1: Neither $\alpha$ nor $\beta$ vertices are included in $V(\mathcal{P}_\nu)$} ($x_\alpha^\nu = 0$ and $x_\beta^\nu = 0$).
    
    In this case, the new edges do not connect to any existing path in $\mathcal{P}_\nu$. Therefore, neither the path set nor its neighborhood changes. Thus, $y_\alpha^\mu=y_\alpha^\nu$,
         $y_\beta^\mu=y_\beta^\nu$,
         $z_\alpha^\mu=z_\alpha^\nu$,
        $z_\beta^\mu=z_\beta^\nu$, and
       $p_{ij}^\mu=p_{ij}^\nu$ for $i,j\in [r]$ hold.

    \item \textbf{Case 2: Only one of the labels appears on paths in $\mathcal{P}_\nu$} (e.g., $x_\alpha^\nu > 0$ and $x_\beta^\nu = 0$).
    
    In this case, every vertex with label $\beta$ (which is not included in $\mathcal{P}_\nu$) becomes adjacent to every vertex with label $\alpha$. This means all vertices that were previously non-adjacent to the solution (the $z_\beta^\nu$ vertices) now become neighbors. This yields $y^\mu_\beta = y^\nu_\beta+z^\nu_\beta$ and $z^\mu_\beta = 0$. Since no vertices with label $\beta$ are included in $\mathcal{P}_\nu$, no new paths are formed by connecting existing ones. Thus, the number of paths for endpoint labels remains the same: $p_{ij}^\mu=p_{ij}^\nu$ for all $i,j$. The case where $x_\beta^\nu > 0$ and $x_\alpha^\nu = 0$ is symmetric.

    \item \textbf{Case 3: Both $\alpha$ and $\beta$ vertices are on paths in $\mathcal{P}_\nu$} ($x_\alpha^\nu \ge 1$ and $x_\beta^\nu \ge 1$).
    
    In this case, all vertices with labels $\alpha$ or $\beta$ that are not already included in $V(\mathcal{P}_\nu)\cup N(V(\mathcal{P}_\nu))$ become neighbors of $V(\mathcal{P}_\nu)$. This yields $y^\mu_\alpha=y^\nu_\alpha+z^\nu_\alpha$, $y^\mu_\beta=y^\nu_\beta+z^\nu_\beta$, and $z^\mu_\alpha = z^\mu_\beta = 0$.
    It remains to update the path count vector $(p_{ij}^\mu)$. Since the newly introduced edges may merge several paths in $\mathcal{P}_\nu$ into longer paths, we must enumerate all possible path count vectors that can be obtained from the path count vector $(p_{ij}^\nu)$ after the join operation. This can be done in time $n^{O(r^2)}$ by \Cref{lem:joindisjointpath}.
\end{itemize}

\Cref{lem:joindisjointpath} is based on the update rule at a join node in the dynamic programming for \textsc{Hamiltonian Path} by Espelage et al.~\cite{wg/EspelageGW01}, and we state the corresponding argument explicitly for completeness.
For a path count vector $(p_{11}^\nu,\dots,p_{rr}^\nu)$ at the child node $\nu$,
we say that a path count vector $(p_{11}^\mu,\dots,p_{rr}^\mu)$ is \emph{feasible} for the join node $\mu$ with respect to $(p_{11}^\nu,\dots,p_{rr}^\nu)$
if there exists a set of vertex-disjoint paths in $G_\mu$ obtained from a set of vertex-disjoint paths represented by $(p_{11}^\nu,\dots,p_{rr}^\nu)$
using only edges added by the join operation, whose path count vector is $(p_{11}^\mu,\dots,p_{rr}^\mu)$.

\begin{lemma}\label{lem:joindisjointpath}
Given a path count vector $(p_{11}^\nu,\dots,p_{rr}^\nu)$ at node $\nu$,
all feasible path count vectors for the join node $\mu$ with respect to $(p_{11}^\nu,\dots,p_{rr}^\nu)$
can be enumerated in time $n^{O(r^2)}$.
\end{lemma}
\begin{proof}
A join node $\mu$ with a child $\nu$ adds edges between all vertices with label $\alpha$ and all vertices with label $\beta$ to the graph $G_\nu$.
Let $\mathbf{p}=(p_{ij})$ denote a path count vector. For labels $a,b\in [r]$, a \emph{connecting operation} decreases $p_{\alpha a}$ and $p_{\beta b}$ by $1$, increases $p_{ab}$ by $1$, and leaves all other entries unchanged; this corresponds to joining two distinct paths of types $\langle \alpha,a\rangle$ and $\langle \beta,b\rangle$ into one path of type $\langle a,b\rangle$. Such an operation is applicable if either $p_{\alpha\beta}\ge 2$ and $(a,b)=(\beta,\alpha)$, or $p_{\alpha a}\ge 1$ and $p_{\beta b}\ge 1$ with $(a,b)\neq(\beta,\alpha)$. This condition follows from the fact that, when connecting two paths by a join edge whose endpoints have labels $\alpha$ and $\beta$, the two paths are either two distinct $\langle \alpha,\beta\rangle$-paths, or one $\langle \alpha,a\rangle$-path and one $\langle \beta,b\rangle$-path.

We claim that a path count vector is feasible for the join node $\mu$ with respect to $(p_{11}^\nu,\dots,p_{rr}^\nu)$ if and only if it can be obtained from $(p_{11}^\nu,\dots,p_{rr}^\nu)$ by a finite sequence of connecting operations. 
The sufficient condition follows from the fact that each connecting operation corresponds to adding one join edge between the endpoints of two distinct paths, thereby merging them into a longer path while preserving vertex-disjointness. Hence any finite sequence of connecting operations produces a feasible path count vector.

For the necessary condition, let $\mathbf{p}$ be a feasible path count vector for the join node $\mu$ with respect to $(p_{11}^\nu,\dots,p_{rr}^\nu)$ and $\mathcal{P}$ be the corresponding set of vertex-disjoint paths in $G_\mu$. We use induction on the number of added edges used to merge paths at the join node. If this number is zero, then the path count vector is exactly $(p_{11}^\nu,\dots,p_{rr}^\nu)$. Otherwise, choose one added edge $\{u,v\}$ on a path $P$, where $u$ has label $\alpha$ and $v$ has label $\beta$. Since $\mathcal{P}$ is a set of paths, removing $\{u,v\}$ splits $P$ into two paths of types $\langle \alpha,a\rangle$ and $\langle \beta,b\rangle$ for some $a,b\in [r]$, which is the reverse of one connecting operation. Let $\mathcal{P}'$ be the set obtained from $\mathcal{P}$ by replacing $P$ with these two paths. Then $\mathcal{P}'$ is again a set of vertex-disjoint paths in $G_\mu$, and the corresponding path count vector is feasible with respect to $(p_{11}^\nu,\dots,p_{rr}^\nu)$, since it is obtained from the same set of paths represented by $(p_{11}^\nu,\dots,p_{rr}^\nu)$ in $G_\nu$ by using one fewer added join edge. The claim follows by induction, and therefore every feasible path count vector can be obtained from $(p_{11}^\nu,\dots,p_{rr}^\nu)$ by a finite sequence of connecting operations.

There are $O(r^2)$ entries in a path count vector, and each entry is an integer between $0$ and $n$. Hence, the total number of path count vectors is at most $n^{O(r^2)}$. Starting from $(p_{11}^\nu,\dots,p_{rr}^\nu)$, we enumerate all reachable vectors by breadth-first search over this state space. From each state, there are only $O(r^2)$ choices of labels $(a,b)$ for a connecting operation, and each applicability test and update takes polynomial time. Therefore, all feasible path count vectors can be enumerated in time $n^{O(r^2)}$.
\end{proof}

From the above argument, 
$\mathrm{DP}_\mu[*_\mu]=\texttt{true}$ if and only if there exist the state tuple $*_{\nu}$  such that $\mathrm{DP}_{\nu}[*_\nu]=\texttt{true}$ and

\begin{itemize}
    \item $x_i^\mu=x_i^\nu$ for $i\in [r]$,
    \item $y_i^\mu=y_i^\nu$ for $i\in [r]\setminus\{\alpha,\beta\}$,
    \item $z_i^\mu=z_i^\nu$ for $i\in [r]\setminus\{\alpha,\beta\}$,
    
    \item if $x_\alpha^\nu=0$ and $x_\beta^\nu=0$,
    \begin{itemize}
        \item $y_\alpha^\mu=y_\alpha^\nu$,
        \item $y_\beta^\mu=y_\beta^\nu$,
        \item $z_\alpha^\mu=z_\alpha^\nu$,
        \item $z_\beta^\mu=z_\beta^\nu$,
        \item $p_{ij}^\mu=p_{ij}^\nu$ for $i,j\in [r]$.
    \end{itemize}
    
    \item if $x_\alpha^\nu\neq0$ and $x_\beta^\nu=0$,
    \begin{itemize}
        \item $y_\alpha^\mu=y_\alpha^\nu$,
        \item $y_\beta^\mu=y_\beta^\nu+z_\beta^\nu$,
        \item $z_\alpha^\mu=z_\alpha^\nu$,
        \item $z_\beta^\mu=0$,
        \item $p_{ij}^\mu=p_{ij}^\nu$ for $i,j\in [r]$.
    \end{itemize}
    
    \item if $x_\alpha^\nu=0, x_\beta^\nu\neq0$
    \begin{itemize}
        \item $y_\alpha^\mu=y_\alpha^\nu+z_\alpha^\nu$
        \item $y_\beta^\mu=y_\beta^\nu$
        \item $z_\alpha^\mu=0$
        \item $z_\beta^\mu=z_\beta^\nu$ 
        \item $p_{ij}^\mu=p_{ij}^\nu$ for $i,j\in [r]$.
    \end{itemize}
    
    \item if $x_\alpha^\nu\neq0$ and $x_\beta^\nu\neq0$,
    \begin{itemize}
        \item $y_\alpha^\mu=y_\alpha^\nu+z_\alpha^\nu$,
        \item $y_\beta^\mu=y_\beta^\nu+z_\beta^\nu$,
        \item $z_\alpha^\mu=0$,
        \item $z_\beta^\mu=0$,
        \item $(p_{11}^\mu,\cdots,p_{rr}^\mu)$ is equal to one of the path count vectors enumerated by \Cref{lem:joindisjointpath} from $(p_{11}^\nu,\cdots,p_{rr}^\nu)$.        
    \end{itemize}
\end{itemize}

We now summarize the overall time complexity. The total running time is the sum of computations over all $O(n)$ nodes in the expression tree $\mathcal{T}$. The bottleneck is the update at a join node $\mu$. To compute the DP table for $\mu$, our algorithm iterates through each of the up to $n^{O(r^2)}$ valid states in the table of its child, $\nu$. For each such state, we perform the enumeration described in \Cref{lem:joindisjointpath}, which takes $n^{O(r^2)}$ time. Therefore, the time to compute the table for a single join node is $ n^{O(r^2)} \cdot n^{O(r^2)} = n^{O(r^2)}$.

The total complexity for the entire algorithm is thus $O(n) \cdot n^{O(r^2)}$, which simplifies to $n^{O(r^2)}$. By substituting $r = 2^{\cw(G)+1}+1$, this running time is $n^{2^{O(\cw(G))}}$. This is an XP algorithm parameterized by cliquewidth. Therefore, \Cref{thm:cw_XP} holds.

\bigskip

By applying \Cref{thm:cw_XP} for each $k,l$, we can also solve \textsc{Secluded $k$-Path} and \textsc{Short Secluded Path}.

\begin{corollary}\label{cor:cw_XP}
\textsc{Secluded $k$-Path} and \textsc{Short Secluded Path} can be solved in time $n^{2^{O(\cw)}}$.
\end{corollary}

\section{FPT Algorithm Parameterized by Neighborhood Diversity}\label{sec:nd}
In this section, we propose a fixed-parameter algorithm for \textsc{Short Secluded Path} parameterized by neighborhood diversity, which runs in $\nd^{O(\nd^2)}n^{O(1)}$ time.

\subsection{FPT algorithm for $s$-$t$ $k$-Path}
We first give an ILP formula with $O(\nd^2)$ variables for determining whether $G$ has an $s$-$t$ path of length $k$ passing through every module.

Let $\mathcal{G}$ be the quotient graph with respect to a twin partition $\mathcal{M}=\{M_1,M_2,\ldots,M_r\}$, where $r$ is the number of modules in the partition of $G$.
Suppose that $M_s=\{s\}$ and $M_t=\{t\}$ for $s,t\in V(G)$ are modules in $\mathcal{M}$, and $\mathcal{G}$ is connected.

We adopt an ILP formulation similar to the one used for \textsc{Hamiltonian Cycle} on graphs of bounded modular-width \cite[Lemma~6]{iwpec/GajarskyLO13}. 
We define the ILP instance \textup{(P)}
with variables $x_{ij},x_{ji}$ (for $\{M_i, M_j\} \in E(\mathcal{G})$), and $y_i$ (for $i\in [r]$) subject to the following constraints.

\begin{tcolorbox}[title=ILP instance \textup{(P)}]

\begin{enumerate}
\setcounter{enumi}{0}
    \item{
    $\sum_{i=1}^r y_i=k$
    }
\end{enumerate}

\begin{enumerate}
\setcounter{enumi}{1}
    \item
    \begin{enumerate}
    \item{
    $\sum_{j\in \lbrace l : M_l\in N_{\mathcal{G}}(M_i)\rbrace }x_{ij}=\sum_{j\in \lbrace l:M_l\in N_{\mathcal{G}}(M_i)\rbrace} x_{ji}$
    }
    \quad for every $M_i\in \mathcal{M}\setminus \{M_s,M_t\}$
    
    \item{
    $\sum_{j\in \lbrace l : M_l\in N_{\mathcal{G}}(M_s)\rbrace }x_{sj}=1$}
    \\
    {\large$\sum_{j\in \lbrace l:M_l\in N_{\mathcal{G}}(M_s)\rbrace} x_{js}=0$
    }
    
    \item{
    $\sum_{j\in \lbrace l:M_l\in N_{\mathcal{G}}(M_t)\rbrace} x_{jt}=1$
    }\\
    {
    $\sum_{j\in \lbrace l : M_l\in N_{\mathcal{G}}(M_t)\rbrace }x_{tj}=0$}
    \end{enumerate}
    
    \item
    \begin{enumerate}
    \item{
     $\sum_{j\in \lbrace l:M_l\in N_{\mathcal{G}}(M_i)\rbrace }x_{ij}=y_i$}
     \quad (if $M_i$ is an independent set)
    \item{
     $1\le\sum_{j\in \lbrace l:M_l\in N_{\mathcal{G}}(M_i)\rbrace }x_{ij}\le y_i$}
     \quad (if $M_i$ is a clique)
    \end{enumerate}

\end{enumerate}

For every partition of $V(\mathcal{G})$ into vertex sets $A$ and $V(\mathcal{G})\setminus A$:
\begin{enumerate}
\setcounter{enumi}{3}
    \item{
    $\sum_{1\le i< j\le r:\lbrace M_i,M_j \rbrace\in E(\mathcal{G})\land |\lbrace M_i,M_j \rbrace\cap A|=1}x_{ij}+x_{ji}\ge1$
    }
\end{enumerate}

 For every variable $x_{ij}$: 
\begin{enumerate}
\setcounter{enumi}{4}
    \item {
    $x_{ij}\ge 0$
    }
\end{enumerate}

 For every variable $y_{i}$: 
\begin{enumerate}
 \setcounter{enumi}{5}
    \item {
    $1\le y_i\le |M_i|$
    }
\end{enumerate}
\end{tcolorbox}


The core idea is based on the standard \emph{flow-conservation principle} commonly used in network flow problems. 
We imagine one unit of flow from a source module $M_s$ to a sink module $M_t$, being conserved at all intermediate modules.
The variables $x_{ij}$ and $x_{ji}$ represent the amounts of flows the $s$-$t$ path goes from $M_i$ to $M_j$ and from $M_j$ to $M_i$, respectively; these correspond to the number of times the $s$-$t$ path traverses from $M_i$ to $M_j$ and from $M_j$ to $M_i$. The integer variable $y_i$ denotes the number of vertices used within module $M_i$ that the path visits.


The constraint (1) guarantees that the size of a solution is exactly $k$.
The constraints (2) are the flow conservation constraints. Constraint (2.a) ensures that for any intermediate module (not a source or sink), the incoming flow equals the outgoing flow, implying that the number of incoming edges and the number of outgoing edges are equal in the intermediate modules. Constraints (2.b) and (2.c) define the source and sink, respectively: exactly one unit of flow leaves the source module $M_s$, and exactly one unit of flow enters the sink module $M_t$. Note that $M_s=\{s\}$ and $M_t=\{t\}$ by assumption.

The constraints (3) govern the relationship between the path passing through a module and the number of vertices used within it ($y_i$).
The constraint (3.a) implies that every incoming edge immediately goes out in an independent set module $M$.
Thus, the number of incoming edges is equal to the number of used vertices in $M$.
The constraint (3.b) implies that $P$ may pass through several vertices and then go out in a clique module $M$.
Thus, the number of incoming edges is at most the number of used vertices in $M$.

The constraint (4) is the connectivity condition, which ensures that a solution forms a single path. This disallows invalid solutions, such as a subgraph composed of a path and disjoint cycles.

The constraint (5) is a non-negativity constraint. The constraint (6) ensures that for each module, a solution contains at least one vertex in the module and no more vertices than the size of the module.

\begin{lemma}\label{lem:ILP}
The ILP instance \textup{(P)} is feasible if and only if there is an $s$-$t$ path in $G$ with $k$ vertices that includes at least one vertex from each module $M_i \in \mathcal{M}$.
\end{lemma}
\begin{proof}
Suppose there is an $s$-$t$ path $P$ in $G$ with $k$ vertices that includes at least one vertex from each $M_i$. We can assume $P$ is a directed path from $s$ to $t$.
For every $\{M_i, M_j\} \in E(\mathcal{G})$, let $x_{ij}$ be the number of arcs $(u,v)$ in $P$ such that $u \in M_i$ and $v \in M_j$.
For every $i \in [r]$, let $y_i = |V(P) \cap M_i|$.
Then the size constraint (1) $\sum_{i=1}^r y_i=k$ holds.

Since $P$ is an $s$-$t$ path that visits each module, the flow conservation constraints (2) are satisfied. For any module $M_i$ other than $M_s$ and $M_t$, the number of edges entering $M_i$ from other modules must equal the number of edges leaving $M_i$. For $M_s$, one edge leaves, and none enter from other modules. For $M_t$, one edge enters, and none leave. Note that $M_s = \{s\}$ and $M_t=\{t\}$.

For constraint (3), if $M_i$ is an independent set, the path cannot use two consecutive vertices from $M_i$. Thus, for each vertex visited in $M_i\in \mathcal{M}\setminus \{M_s,M_t\}$, the path must enter and then immediately leave the module. So, the number of entries equals the number of vertices visited, $\sum_j x_{ij} = y_i$. If $M_i$ is a clique, the path can traverse multiple vertices within $M_i$ after entering once. Thus, the number of entries must be at least 1 (if $y_i > 0$) and at most $y_i$.

Constraint (4) is the cut constraint, ensuring the path is connected from $s$ to $t$. Since $P$ is a connected $s$-$t$ path, for any cut separating $s$ and $t$, at least one arc must cross the cut.

By their definitions, each $x_{ij}$ is non-negative, so constraint (5) holds. By assumption, $P$ visits at least one vertex in each module, and $y_i$ cannot exceed the module's size, so constraint (6) holds. Thus, \textup{(P)} is feasible.

For the reverse direction, suppose the ILP instance \textup{(P)} has a feasible solution $x_{ij}, y_i$. Let $\mathcal{G}'$ be the directed multigraph on the vertex set $\mathcal{M}$ where for each $\{M_i, M_j\} \in E(\mathcal{G})$, we add $x_{ij}$ arcs from $M_i$ to $M_j$ and $x_{ji}$ arcs from $M_j$ to $M_i$. Constraint (4) guarantees the connectivity of $\mathcal{G}'$. Constraint (2)  ensures that $\mathcal{G}'$ has an Eulerian trail from $M_s$ to $M_t$, that is, a directed walk from $M_s$ to $M_t$ that traverses every edge exactly once in  $\mathcal{G}'$. This trail defines an ordering $\pi$ of the arcs in $\mathcal{G}'$.

For each module $M_i$, let $n_i = \sum_{M_j \in N_{\mathcal{G}}(M_i)} x_{ij}$ be the number of times the trail $T$ leaves $M_i$. We construct a path in $G$ as follows. For each $M_i$, we select $y_i$ vertices to be on the path.
If $M_i$ is an independent set, constraint (3) implies $y_i = n_i$. We select any $y_i$ vertices and partition them into $n_i$ paths of length 1 (single vertices).
If $M_i$ is a clique, we select any $y_i$ vertices. We can partition these $y_i$ vertices into $n_i$ paths. For example, $n_i-1$ paths of length 1 and one path containing the remaining $y_i - n_i + 1$ vertices. This is possible since $M_i$ is a clique.

Let $\mathcal{P}_{M_i} = \{P_1^{(i)}, \dots, P_{n_i}^{(i)}\}$ be the set of these disjoint paths for each $M_i$. We stitch these paths together according to the Eulerian trail $T$. For each arc $a=(M_i, M_j)$ in $T$, we add an edge in $G$ connecting the endpoint of a path in $\mathcal{P}_{M_i}$ to the startpoint of a path in $\mathcal{P}_{M_j}$. Since $M_i$ and $M_j$ are adjacent in $\mathcal{G}$, such an edge exists between any vertex in $M_i$ and any vertex in $M_j$. We use each path in $\mathcal{P}_{M_i}$ exactly once as a source and once as a destination (except for $s$ and $t$). This process constructs a simple $s$-$t$ path $P$ in $G$.

The total number of vertices in $P$ is $\sum y_i$, which equals $k$ by constraint (1). By constraint (6), at least one vertex is chosen from each module. Thus, we have constructed the required $s$-$t$ $k$-path.
\end{proof}

From Lemma \ref{lem:ILP} and Theorem \ref{thm:ILP_proposition}, we obtain the following lemma.

\begin{lemma}\label{lem:st-k-path-fpt}
Given a twin partition with $r$ modules, the problem of determining whether there is an $s$-$t$ path of length $k$ passing through every module can be solved in $r^{O(r^2)}n^{O(1)}$ time.
\end{lemma}
\begin{proof}
By Lemma \ref{lem:ILP}, the problem can be formulated as the $p$-ILP problem. The number of variables is $p = O(r^2)$. The number of constraints is dominated by the $O(2^r)$ cut constraints (Constraint 4).
Since the size of each module is at most $n$, the size of the input is $O(r^2 2^r + r\log n)$.
By Theorem~\ref{thm:ILP_proposition}, the problem can be solved in time $p^{O(p)} \cdot L = (r^2)^{O(r^2)} (r^2 2^r + r\log n) = r^{O(r^2)}n^{O(1)}$. 
\end{proof}

\subsection{FPT algorithm for \textsc{Short Secluded Path}}
Using \Cref{lem:st-k-path-fpt}, we present an FPT algorithm for \textsc{Secluded $k$-Path} and \textsc{Short Secluded Path}  parameterized by neighborhood diversity.
A key observation is that if a path visits a module $M$, then all vertices from this module not on the path are in the neighborhood of the path. Also, all the vertices in modules adjacent to $M$ are in the neighborhood of the path. Therefore, what we only have to do is to solve the problem of determining whether there is an $s$-$t$ path of length $k$ passing through every guessed module by using \Cref{lem:st-k-path-fpt}.

\begin{theorem}\label{thm:nd_fpt:kpath}
\textsc{Secluded $k$-Path} parameterized by neighborhood diversity $\nd(G)$ can be solved in $\nd(G)^{O(\nd(G)^2)}n^{O(1)}$ time.
\end{theorem}
\begin{proof}
First, we compute a twin partition for $G$ with $\nd(G)$ modules in linear time~\cite{Lampis12,McConnell,Tedder}.
We ensure that $s$ and $t$ are in their own singleton modules by splitting the modules containing them. If $s \in M$, we replace $M$ with $M\setminus \{s\}$ and add a new module $M_s=\{s\}$. We do the same for $t$. This results in a new twin partition $\mathcal{M}=\lbrace M_1,M_2,\ldots, M_{r} \rbrace$ with $r \le \nd(G)+2$.
The algorithm proceeds by guessing which modules $\mathcal{M}' \subseteq \mathcal{M}$ contain the vertices of a potential $l$-secluded $s$-$t$ path of length at most $k$. We must always include $M_s$ and $M_t$ in $\mathcal{M}'$. There are $O(2^{r-2}) = O(2^{\nd})$ possible choices for $\mathcal{M}'$. 

Let $r' = |\mathcal{M}'|$.
For each guess, we check if there exists a valid solution containing at least one vertex in each module in $\mathcal{M}'$ in the subgraph induced by the vertices of $\mathcal{M}'$. 
If $|r'|>k$ or the induced subgraph is not connected, we can discard this guess.

Let $\mathcal{N} \subseteq \mathcal{M} \setminus \mathcal{M}'$ be the set of modules adjacent to at least one module in $\mathcal{M}'$. For any $s$-$t$ path $P$ with $V(P) \subseteq \bigcup_{M \in \mathcal{M}'} M$, its neighborhood $N(V(P))$ consists of two parts: (1) all vertices in the modules of $\mathcal{N}$, and (2) all vertices in modules of $\mathcal{M}'$ that are not on the path $P$. 
Therefore, $N(V(P)) = (\bigcup_{M \in \mathcal{N}} M) \cup ((\bigcup_{M \in \mathcal{M}'} M) \setminus V(P))$.
The size of the neighborhood is $|N(V(P))| = |\bigcup_{M \in \mathcal{N}} M| + |\bigcup_{M \in \mathcal{M}'} M| - |V(P)|$.
Since $|V(P)|=k$, $|N(V(P))|$ is given by $|\bigcup_{M \in \mathcal{N}} M| + |\bigcup_{M \in \mathcal{M}'} M| - k$ for a fixed guess $\mathcal{M}'$ (which fixes $\mathcal{N}$). Thus, if $|N(V(P))|>l$, we discard such a guess.


Finally, we only have to determine whether there exists a path of length $k$
in the graph $G[\bigcup_{M' \in \mathcal{M}'} M']$ that visits every module in $\mathcal{M}'$.
By \Cref{lem:st-k-path-fpt}, this can be determined in time $r'^{O(r'^2)}n^{O(1)} = \nd^{O(\nd^2)}n^{O(1)}$.
We repeat this for all $O(2^{\nd})$ guesses of $\mathcal{M}'$. The total runtime is $2^{\nd}  \cdot \nd^{O(\nd^2)}n^{O(1)}=\nd^{O(\nd^2)}n^{O(1)}$.
\end{proof}

By applying Theorem~\ref{thm:nd_fpt:kpath} to \textsc{Secluded $k'$-Path} for $2\le k'\le k$, \textsc{Short Secluded Path} can also be solved in $\nd^{O(\nd^2)}n^{O(1)}$ time. 
\begin{corollary}\label{cor:SSP:nd}
For a graph $G$ of neighborhood diversity $\nd$, \textsc{Short Secluded Path} can be solved in $\nd^{O(\nd^2)}n^{O(1)}$ time.
\end{corollary}

\section{FPT Algorithm Parameterized by Twin Cover Number}\label{sec:tc}

In this section, we design a $2^{O(\tc^2)}n^{O(1)}$-time algorithm for \textsc{Short Secluded Path} and \textsc{Secluded $k$-Path} parameterized by twin cover number.

\begin{theorem}\label{thm:tc_fpt}
\textsc{Secluded $k$-Path} can be solved in $2^{O(\tc(G)^2)}n^{O(1)}$ time.
\end{theorem}

\begin{proof} 
For any graph $G$, a twin cover $S$ of size $\tc$ can be computed in $O(m+n\cdot\tc+{1.2738}^{\tc})$ time~\cite{dmtcs/Ganian15}.
Let $X := S \cup \{s,t\}$ and $\tau := |X| (=\tc+2)$.

 First, we guess the set of edges from $G[X]$ that are part of the secluded path. The number of edges in $G[X]$ is at most $\binom{\tau}{2} \le \tau^2$. Since a path on $\tau$ vertices has at most $\tau-1$ edges, we guess subsets of at most $\tau-1$ edges as candidates for partial solutions inside of $X$. The number of guessed subsets is bounded by ${\sum_{i=0}^{\tau-1}\binom{\tau^2}{i} = \tau^{O(\tau)}}$.

For each guess, we check if the subgraph formed by these edges consists of a collection of disjoint paths, if $s$ and $t$ are part of the subgraph because they must be endpoints of a path, and if the sum of the number of vertices in the guessed disjoint paths is at most $k$. This clearly can be done in polynomial time. If these conditions are not satisfied, we discard the current guess.

Let $P_1, P_2, \ldots, P_d$ be  sub-paths within $G[X]$ of a guessed partial solution, where $d \le \tau$. Let $X' = \bigcup_{i=1}^d V(P_i)$ be the set of vertices in these paths.

Next, we construct a full $s$-$t$ path by ordering these sub-paths and connecting them with cliques from $G[V\setminus X]$.
First, we guess the orderings of the sub-paths.
If $s$ and $t$ belong to the same sub-path $P_i$, then this must be the entire path. In this case, we check if $d=1$, $|X'| = |V(P_1)| = k$, and $|N(V(P_1))|\le l$. If so, $P_1$ is indeed an $l$-secluded $s$-$t$ path of length $k$, and otherwise we safely conclude that the guess is invalid.
If $s$ and $t$ are in different sub-paths, we fix the path containing $s$ at the beginning and the one containing $t$ at the end. We then guess the ordering of the remaining $d-2$ paths.
Since $d-2\le \tau$, the number of such orderings is $(d-2)! =$ $\tau^{O(\tau)}$.

For a fixed ordering of sub-paths, we need to connect them. Since $X$ is a separator and $G[V\setminus X]$ is a union of disjoint cliques, vertices in exactly one clique connect two sub-paths in $G[X]$.

Here, we categorize cliques in $G[V\setminus X]$ with respect to neighbors in $X$, and we define the sets of vertices $Y_1, \ldots, Y_q$ where vertices in $Y_i$ have the common neighbors in $X$, that is, $N(u)\cap X = N(v)\cap X$ for $u,v\in Y_i$. Note that $q \le 2^\tau$. By the definition of a twin cover, each connected component of $G[V\setminus X]$ is a clique and is fully contained in some $Y_i$. We say that a clique $C$ in $G[V\setminus X]$ is from $Y_i$ if $V(C)\subseteq Y_i$.

To connect the $d-1$ gaps between the ordered sub-paths, we guess which set $Y_j$ provides the connecting clique for each gap. Since there are $d-1 = O(\tau)$ gaps and $q=O(2^\tau)$ choices for each, the total number of ways to choose the sequence of $Y_j$'s is $q^{d-1}= 2^{O(\tau^2)}$.

For each such guess, we verify whether it is valid. If a chosen $Y_j$ cannot connect the corresponding sub-paths, we reject it.
We also calculate the minimum possible path length. Since we must use at least one vertex from a clique in each of the $d-1$ gaps, the length $p$ of an $s$-$t$ path constructed is at least $p \ge \sum^d_{i=1}|V(P_i)|+d-1$.
If $p>k$, we safely reject it.
Furthermore, if the number of times a set $Y_i$ is selected is more than $|Y_i|$, then we reject it because we cannot construct an $s$-$t$ path even by using all the vertices in $Y_i$.

Now, for each gap between two sub-paths, we must select a specific clique from the chosen $Y_i$. Then by the following claim, we can focus on at most $r$ largest cliques in each $Y_i$.

\begin{claim}\label{tc_clique}
Suppose that $Y_i$ contains $r$ cliques $C_1, \ldots, C_r$ ordered by non-increasing size ($|C_1| \ge \ldots \ge |C_r|$). Then if there exists an $s$-$t$ path $P$ that passes through $b$ cliques from $Y_i$, there exists an $s$-$t$ path $P'$ that passes through the $b$ largest cliques from $Y_i$, $C_1, \ldots, C_b$ (or all $r$ cliques if $b \ge r$), such that $|V(P)|=|V(P')|$ and $|N(V(P))|=|N(V(P'))|$.
\end{claim} 

\begin{claimproof}
Suppose that there exists a clique $C$ among the $b$ largest cliques from $Y_i$ not used by $P$. Without loss of generality, we assume that $C$ is the largest clique among the $b$ largest cliques from $Y_i$ not used by $P$. Then $P$ uses another clique from $Y_i$ of size at most $|C|$.
Let $P_1,\ldots,P_{d}$ be the ordered sub-paths of $P$ in $X$. When a clique from $Y_i$ connects $P_j$ and $P_{j+1}$, its vertices must be adjacent to the endpoints of $P_j$ and $P_{j+1}$. Since all vertices in $Y_i$ have the same neighborhood in $X$, any clique from $Y_i$ can make this connection.
To construct $P'$, we simply replace the clique $C'$ used by $P$ with $C$, that is, if $P$ passes through $c'$ vertices in $C'$ between $P_j$ and $P_{j+1}$, $P'$ passes through arbitrary $c'$ vertices in $C$ between $P_j$ and $P_{j+1}$. Since 
$|C|\ge |C'|\ge c'$, this replacement is possible.
It is clear that $|V(P')|=|V(P)|$ holds.

We show that $P'$ satisfies $|N(V(P))|=|N(V(P'))|$.
Since $V(P)\cap X=V(P')\cap X$ and any vertex in $Y_i$ has the same neighbors in $X$, it holds that $|N(V(P))\cap X|=|N(V(P'))\cap X|$.
The neighborhood outside $X$ depends on which $Y_j$ sets are used and how many vertices are taken from them. Since we only swap cliques within the same $Y_i$ and keep the number of vertices taken from $Y_i$, i.e., $|V(P) \cap Y_i| = |V(P') \cap Y_i|$ holds, we have $|N(V(P)\cap Y_j)| = |N(V(P')\cap Y_j)|$ for any $j$.

Then, the following equation holds:
\begin{align*}
|N(V(P))\setminus X|=\sum_{j=1}^{q} |N(V(P))\cap Y_j|=\sum_{j=1}^{q} |N(V(P'))\cap Y_j|=|N(V(P'))\setminus X|.
\end{align*}
Therefore, the total neighborhood size is preserved. By repeating this replacement, we can obtain an $s$-$t$ path $P'$ that passes through the $b$ largest cliques, $C_1, \ldots, C_b$ (or all $r$ cliques if $b \ge r$), such that $|V(P)|=|V(P')|$ and $|N(V(P))|=|N(V(P'))|$.
\end{claimproof}

Now we need to decide vertices to take from the selected cliques. Here, we observe that the vertex set $V$ can be partitioned into four parts: (1) the set $V(P)\cap X$ of vertices on $P$ inside $X$, (2) the set of vertices $\bigcup_{C\in \mathcal{C}^P}C$ in the selected cliques in $G[V\setminus X]$, (3) the set $N\left((V(P) \cap X)\cup \bigcup_{C\in \mathcal{C}^P}C\right)$ of neighbors of $P$ except for vertices in $\bigcup_{C\in \mathcal{C}^P}C$, and (4) The set $V\setminus V(P)\cup N(V(P))$ of the other vertices which lie outside both $X$ and $N(V(P))$ (see \Cref{fig:tc_fpt_claim_4_2}).
Then the following claim holds.

\begin{figure}[tbp]
\centering
\includegraphics[width=\linewidth]{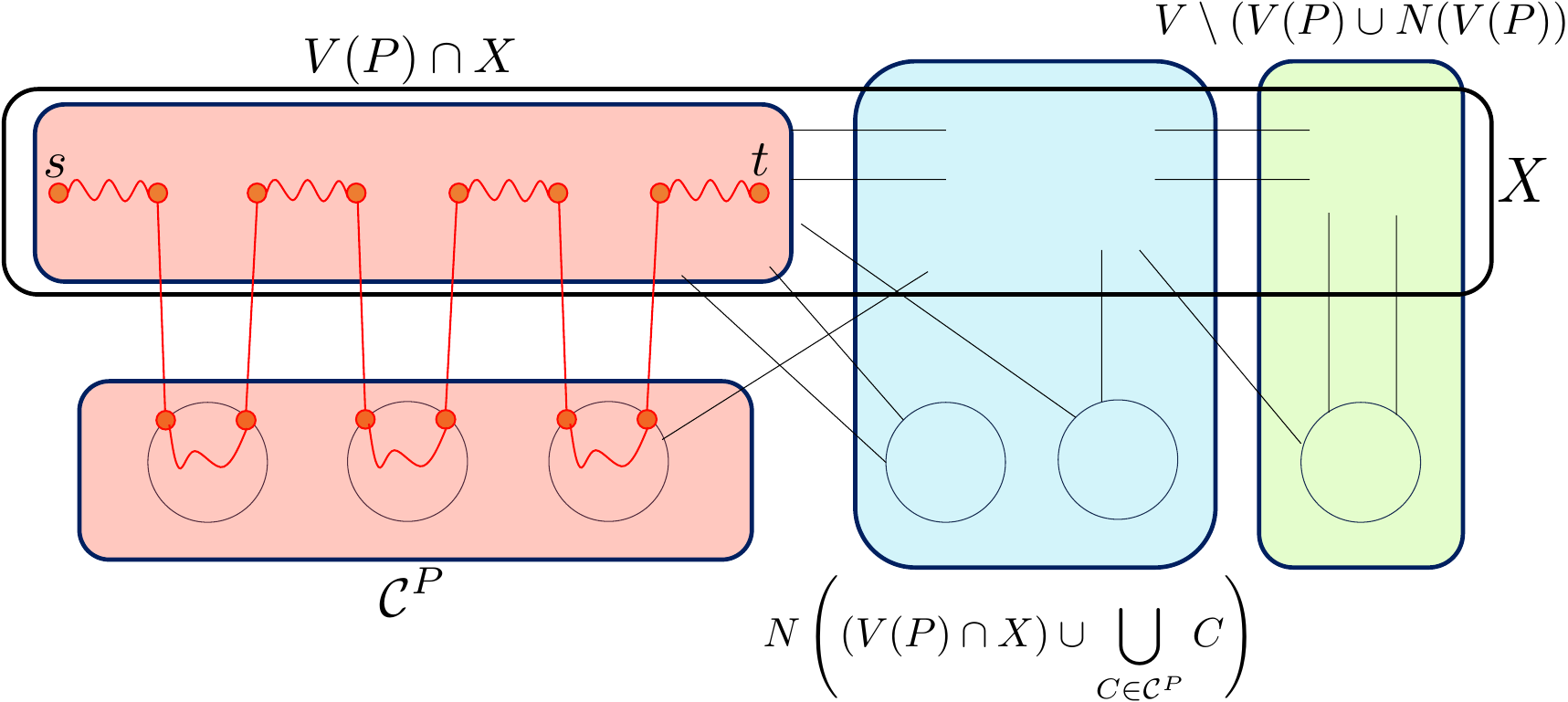}
\caption{An illustration of a partition of $V$ with respect to a path $P$ in Claim \ref{tc_vertex}.}
\label{fig:tc_fpt_claim_4_2}
\end{figure}

\begin{claim}\label{tc_vertex}
Let $P$ be an $s$-$t$ path and let $\mathcal{C}^P$ be the set of cliques in $V \setminus X$ that intersect with $V(P)$. The number of neighbors of the path $P$ is given by:
\begin{align}\label{eq:tc:neighbors}
    |N(V(P))|=\left|N\left((V(P) \cap X)\cup \bigcup_{C\in \mathcal{C}^P}C\right)\right|+\left|\bigcup_{C\in \mathcal{C}^P}C\right| + |V(P)\cap X|-|V(P)|.
\end{align}
\end{claim}
\begin{claimproof}
Since every vertex in $C\in\mathcal{C}^P$ has the same neighborhood in $X$ and $P$ only passes through cliques in $\mathcal{C}^P$, the following formula holds.
\begin{align*}
|N(V(P))|&=\left|N\left((V(P) \cap X)\cup\bigcup_{C\in\mathcal{C}^P}C\right)\right|+\left|\bigcup_{C\in\mathcal{C}^P}C\setminus V(P)\right|\\
&=
\left|N\left((V(P) \cap X)\cup \bigcup_{C\in \mathcal{C}^P}C\right)\right|+\left|\bigcup_{C\in \mathcal{C}^P}C\right|-\left|V(P) \cap \bigcup_{C\in \mathcal{C}^P}C\right|\\
&=\left|N\left((V(P) \cap X)\cup \bigcup_{C\in \mathcal{C}^P}C\right)\right|+\left|\bigcup_{C\in \mathcal{C}^P}C\right| + |V(P)\cap X|-|V(P)|.
\end{align*}
This holds as claimed.
\end{claimproof}

If we have fixed the sub-paths in $X$ and the cliques connecting them, we observe that the first three terms $\left|N\left((V(P) \cap X)\cup \bigcup_{C\in \mathcal{C}^P}C\right)\cap V(P)\right|$, $\left|\bigcup_{C\in \mathcal{C}^P}C\right|$, and $|V(P)\cap X|$ in the right-hand side of \Cref{eq:tc:neighbors} are determined. Note that $|V(P)\cap X|=|X'| = \sum_{i=1}^d |V(P_i)|$. Since $|V(P)|=k$, from Claim \ref{tc_vertex}, $|N(V(P))|$ is also determined.

After fixing the sub-paths in $X$ and the $Y_i$'s for connecting pairs of sub-paths, we can identify the set $\mathcal{C}$ of candidate cliques $\mathcal{C}$ (the largest ones, by Claim \ref{tc_clique}).
Since each clique in $Y_i$ have the same neighbors by the definition of a twin cover, if $|\bigcup_{C\in \mathcal{C}} C| = \sum_{C\in \mathcal{C}}|C|\ge k - |V(P)\cap X|$ is satisfied, we can obtain a path of length $k$ (we just need to arbitrarily choose $k - |V(P)\cap X|$ vertices from selected cliques in $\mathcal{C}$ such that at least one vertex is chosen from each selected clique). This is clearly done in polynomial time.

Finally, we analyze the running time. Guessing the sub-paths in $X$, the orderings of the sub-paths, and the $Y_i$'s for connecting pairs of sub-paths can be done in time $\tau^{O(\tau)}n^{O(1)}\cdot 2^{O(\tau^2)}n^{O(1)} =  2^{O(\tau^2)}n^{O(1)}$. After the guesses, we can check whether each guess is valid by verifying whether $\sum_{C\in \mathcal{C}}|C|\ge k - |V(P)\cap X|$ in polynomial time. Therefore, the total time is $2^{O(\tau^2)}n^{O(1)} = 2^{O(\tc^2)}n^{O(1)}$.
\end{proof}

\begin{corollary}\label{cor:SSP:tc}
   \textsc{Short Secluded Path} can be solved in $2^{O(\tc(G)^2)}n^{O(1)}$ time.
\end{corollary}

\section{Computational Complexity of \textsc{Shortest Secluded Path}}\label{sec:shortest}
In this section, we study the computational complexity of \textsc{Shortest Secluded Path}.

\subsection{Polynomial-time algorithm on unweighted graphs}

In this subsection, we present a polynomial-time algorithm for \textsc{Shortest Secluded Path} on unweighted graphs.

Let $k=\dist(s,t)$ be the distance between $s$ and $t$. For $i\in \{0,1,\ldots,k\}$, we define the layer $L_i = \{v\in V\mid \dist(s,v) = i$ and $\dist(v,t) = k-i\}$ as the set of vertices whose distance from $s$ is $i$ and from $t$ is $k-i$. These layers $L_0, \dots, L_k$ can be computed in linear time by performing a breadth-first search from $s$ and from $t$, respectively.
By definition, $L_0 = \{s\}$ and $L_k = \{t\}$. Let $L = \bigcup_i L_i$ and $R = V\setminus L$. It is easily seen that every vertex $v$ satisfying $\dist(s,v) + \dist(v,t)=k$ is contained in $L$, and any vertex in $R$ satisfies that $\dist(s,v) + \dist(v,t)>k$.
Thus, every shortest $s$-$t$ path is contained in $G[L]$. Then, the following lemmas hold.

\begin{lemma}\label{lem:L_i_neighbor}
For any integer $i \in \{0,1,\dots,k\}$, if a vertex in $L_i$ has a neighbor in $L$,
then that neighbor belongs to $L_{i-1}\cup L_i\cup L_{i+1}$.
Equivalently, since $N(L_i)$ excludes $L_i$ itself,
\[
N(L_i)\cap L \subseteq L_{i-1}\cup L_{i+1}.
\]
(Here, we define $L_{-1}=L_{k+1}=\emptyset$.)
\end{lemma}

\begin{proof}
    Let $v_i\in L_i$, and let $v_j\in L_j$ be a neighbor of $v_i$ with $v_j\notin L_i$.
    By the triangle inequality, we must have $\dist(s,v_i) \le \dist(s,v_j) + \dist(v_j,v_i)$, which means $i \le j+1$. Symmetrically, $j \le i+1$.
    Together, these imply that $j$ satisfies $|i-j| \le 1$.
    Since $N(L_i)$ excludes vertices of $L_i$ itself, it follows that
    $N(L_i)\cap L \subseteq L_{i-1}\cup L_{i+1}$.
\end{proof}

\begin{lemma}\label{lem:r_neighbor}
For any vertex $r \in R$ and any integer $i \in \{0, 1, \dots, k\}$, if $r$ has a neighbor in $L_i$, then its set of neighbors in $L$, denoted by $N(r) \cap L$, must satisfy one of the following three conditions:
(1) $N(r) \cap L \subseteq L_i$, (2) $N(r) \cap L \subseteq L_{i-1} \cup L_i$, and (3) $N(r) \cap L \subseteq L_i \cup L_{i+1}$.
\end{lemma}

\begin{proof}
The proof consists of two main claims.

\begin{claim}\label{claim:poly:1}
    If a vertex $r \in R$ has a neighbor in $L_i$, then all of its other neighbors in $L$ must be in $L_{i-1} \cup L_i \cup L_{i+1}$.
\end{claim}
\begin{claimproof}
Let $v_i \in L_i$ be a neighbor of $r$. Let $v_j\in L_j$ be another neighbor of $r$. By definition of $R$, we know $\dist(s,r) + \dist(r,t) > k$.
Since $r$ is adjacent to $v_i$ and $v_j$, the triangle inequality gives us:
\begin{align*}
    \dist(s,r) &\le \dist(s,v_i) + \dist(v_i,r) = i+1, \\
    \dist(r,t) &\le \dist(r,v_j) + \dist(v_j,t) = 1 + (k-j) = k-j+1.
\end{align*}
Summing these inequalities yields $\dist(s,r) + \dist(r,t) \le k + i - j + 2$.
Combining this with $\dist(s,r) + \dist(r,t) > k$, we get $k < k + i - j + 2$, which simplifies to $j < i+2$.
By a symmetric argument, we obtain $i < j+2$.
Together, these inequalities imply $i-2 < j < i+2$, which means $j \in \{i-1, i, i+1\}$. Thus, any neighbor of $r$ in $L$ must be in $L_{i-1} \cup L_i \cup L_{i+1}$.
\end{claimproof}

\begin{claim}\label{claim:poly:2}
    A vertex $r \in R$ cannot have neighbors in both $L_{i-1}$ and $L_{i+1}$ simultaneously.
\end{claim}
\begin{claimproof}
Suppose for contradiction that $r \in R$ has a neighbor $v_{i-1} \in L_{i-1}$ and another neighbor $v_{i+1} \in L_{i+1}$.
Using the triangle inequality:
\begin{align*}
    \dist(s,r) &\le \dist(s,v_{i-1}) + \dist(v_{i-1},r) = (i-1)+1 = i, \\
    \dist(r,t) &\le \dist(r,v_{i+1}) + \dist(v_{i+1},t) = 1 + (k-(i+1)) = k-i.
\end{align*}
Summing these gives $\dist(s,r) + \dist(r,t) \le i + (k-i) = k$.
This contradicts the fact that $r \in R$, which requires $\dist(s,r) + \dist(r,t) > k$.
Therefore, $r$ cannot be adjacent to vertices in $L_{i-1}$ and $L_{i+1}$ simultaneously.
\end{claimproof}

\Cref{claim:poly:1} establishes that if $r$ has a neighbor in $L_i$, its neighbors in $L$ are restricted to at most three adjacent layers: $L_{i-1}, L_i, L_{i+1}$. \Cref{claim:poly:2} further proves that $r$ cannot have neighbors in both the layers, $L_{i-1}$ and $L_{i+1}$, simultaneously. Thus, the neighbors of $r$ in $L$ must be contained in one of the three listed combinations. This completes the proof of the lemma.
\end{proof}

From \Cref{lem:L_i_neighbor,lem:r_neighbor}, the neighbors of vertices in $L_i$ are in $L_{i-1} \cup L_i \cup L_{i+1}\cup R$. Moreover, the neighbors in $R$ are also adjacent to only $L_{i-1} \cup L_i \cup L_{i+1}\cup R$.
This locality is the key to our dynamic programming approach, as it ensures that when extending a path to layer $L_{i+1}$, the change in the neighborhood of the path only depends on the most recent layers.

Our dynamic programming table, $\mathrm{DP}$, stores values for pairs of adjacent vertices $(u_i, u_{i+1})$ where $u_i \in L_i$ and $u_{i+1} \in L_{i+1}$. The entry $\mathrm{DP}[u_i, u_{i+1}]$ is defined as the minimum size of the neighborhood of a shortest $s$-$u_{i+1}$ path that uses the edge $\{u_i, u_{i+1}\}$ as its last edge. If no such path exists, or if $\{u_i, u_{i+1}\} \notin E$, then $\mathrm{DP}[u_i, u_{i+1}] = \infty$.
The solution to the overall problem is then $\min_{u_{k-1} \in L_{k-1}} \mathrm{DP}[u_{k-1}, t]$.

For $i=0$ (the base case), we consider paths of length 2 from $s$ to a vertex $u_1 \in L_1$. Since $\dist(s,u_1)=1$, the edge $\{s,u_1\}$ must exist. The path consists of edge $\{s, u_1\}$. The size of its neighborhood is $|(N(s) \cup N(u_1)) \setminus \{s, u_1\}|$. Thus, we initialize $\mathrm{DP}[s,u_1] = |N(s)\cup N(u_1)|-2$.

For $1 \le i < k$, the recursive formula for $\mathrm{DP}[u_i, u_{i+1}]$ is computed as follows:
$$ \mathrm{DP}[u_i,u_{i+1}] =
   \begin{cases}
       \infty & \text{if } \{u_i, u_{i+1}\} \notin E, \\
       \displaystyle \min_{u_{i-1} \in L_{i-1}} \left\{ \mathrm{DP}[u_{i-1},u_i] + |N(u_{i+1}) \setminus (N(u_{i-1}) \cup N(u_i))| \right\} & \text{otherwise}.
   \end{cases}
$$
Here, $|N(u_{i+1}) \setminus (N(u_{i-1}) \cup N(u_i))|$ represents the number of new neighbors added when extending the path from $u_i$ to $u_{i+1}$. Due to the locality shown by \Cref{lem:L_i_neighbor,lem:r_neighbor}, $u_{i+1}$ never shares neighbors of $\bigcup_{j=0}^{i-2}L_j$. Thus, these new neighbors are precisely those in $N(u_{i+1})$ that are not already neighbors of the path ending in $\{u_{i-1}, u_i\}$. 

Finally,  we analyze the running time. The sets $L_i$ can be computed in $O(m+n)$ time.
The size of the DP table is $O(n^2)$ since the number of layers is at most $n$. For three vertices $u,v,w$, 
$|N(u)\cap(N(v)\cup N(w))|$ can be computed in polynomial time in advance.
Thus, each entry $\mathrm{DP}[u_{i},u_{i+1}]$ can be computed in $O(n)$ time from $\mathrm{DP}[u_{i-1},u_{i}]$.
Therefore, the total running time is $O(n^3)$.

\begin{theorem}\label{thm:unweighted:W[1]}
    \textsc{Shortest Secluded Path} can be solved in $O(n^3)$ time.
\end{theorem}

\subsection{\textsc{Shortest Secluded Path} on weighted graphs}
In this subsection, we discuss \textsc{Shortest Secluded Path} on edge-weighted graphs. We assume that each edge weight is a positive integer.
We first show \textsc{Shortest Secluded Path} on weighted graphs is W[1]-hard with respect to the shortest path distance between $s$ and $t$ by a reduction from \textsc{Multicolored Clique}. Here, the shortest path distance between $s$ and $t$ is defined by the minimum weight of an $s$-$t$ path. 

\problemdef{\textsc{Multicolored Clique}} {A graph $G=(V_1\cup \cdots \cup V_k, E)$ with $n$ vertices where each $V_i$ forms an independent set.}
{Determine whether $G$ has a clique of size $k$.}
\textsc{Multicolored Clique} is W[1]-hard when parameterized by $k$ even on regular graphs~\cite{tcs/FellowsHRV09}.
\begin{theorem}\label{thm:weighted:W[1]}
    \textsc{Shortest Secluded Path} on edge-weighted graphs is W[1]-hard parameterized by the shortest path distance $d$ between $s$ and $t$ in the input graph.
\end{theorem}

\begin{figure}[htbp]
\centering
\includegraphics[scale=0.7]{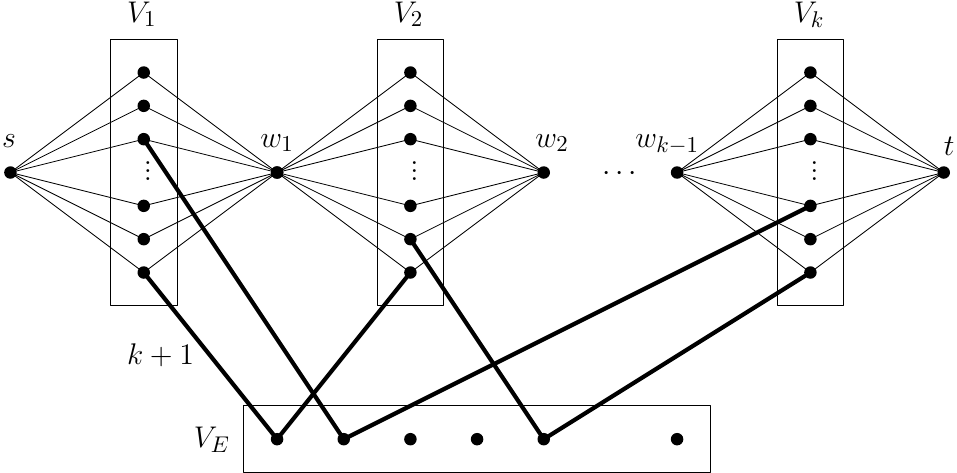}
\caption{The reduction from \textsc{Multicolored Clique} to \textsc{Shortest Secluded Path} in \Cref{thm:weighted:W[1]}. The weight of each thin edge is $1$ and the weight of each bold edge is $k+1$.}
\label{Reduction_MC_SSP}
\end{figure}

\begin{proof}
    We show the parameterized reduction from \textsc{Multicolored Clique} to \textsc{Shortest Secluded Path}.
    
    Let an $r$-regular graph $G=(V,E)$ be an instance of \textsc{Multicolored Clique}, where $V=V_1\cup V_2 \cup\cdots \cup V_k$.
    Then, we construct an equivalent instance $G'=(V',E',s,t)$ of \textsc{Shortest Secluded Path} on edge-weighted graphs.
    We first define $V' = V\cup \{s,t\}\cup V_E\cup W$ where $V_E = \{v_e\mid e\in E\}$ and $W = \{w_h\mid h\in [k-1]\}$.
    Then we connect $s$ to all vertices in $V_1$ and $t$ to all vertices in $V_k$ with edges of weight $1$. Moreover, for each $h\in [k-1]$, we connect $w_h$ to all vertices in $V_h\cup V_{h+1}$ with edges of weight $1$.
    Finally, for $e= \{u,v\}\in E$, we connect $v_e$ to $u$ and $v$ with edges of weight $k+1$. \Cref{Reduction_MC_SSP} illustrates the constructed graph $G'$.
    
    Clearly, $G'$ can be constructed in polynomial time. By the construction of $G'$, the shortest path distance between $s$ and $t$ is $2k$. 
    In the following, we show that there is a clique of size $k$ in $G$ if and only if there is a shortest $s$-$t$ path $P$ of weight $2k$ satisfying $|N(V'(P))|\le rk-\binom{k}{2}+n-k$ in $G'$, where $V'(P)$ is the set of vertices of $P$ in $G'$.
   
    Let $C=\{v_1,v_2,\ldots,v_k\}$ be a clique of size $k$ in $G$, where $v_i\in V_i$.
    We define an $s$-$t$ path $P=\langle s,v_1,w_1,v_2,w_2,\ldots,w_{k-1},v_k,t\rangle$ of weight $2k$.
    Since all edges in $G'[V'\setminus V_E]$ have weight $1$, the weight of $P$ is $2k$, and thus $P$ is a shortest $s$-$t$ path in $G'$.
   
    We next show that $P$ satisfies $|N(V'(P))|\le rk-\binom{k}{2}+n-k$.
    We observe that $V\setminus V'(P)\subseteq N(V'(P))$ holds.
    Since $|V'(P)\cap V|=k$, the neighborhood $N(V'(P))$ of $P$ includes $n-k$ vertices in $V\setminus V'(P)$.
    Furthermore, since $G$ is an $r$-regular graph, any vertex in $V$ has $r$ neighbors in $V_E$ in $G'$.
    Since $C$ is a clique in $G$, for $1\le i<j\le k$, any pair $v_i,v_j\in C$ has an edge $\{v_i,v_j\}$.
    Thus, $v_i$ and $v_j$ have a common neighbor $v_{\{v_i,v_j\}}\in V_E$.
    Therefore, $P$ has only $rk-\binom{k}{2}$ neighbors in $V_E$.
    Consequently, $|N(V'(P))|=rk-\binom{k}{2} + n-k$ holds.

    Conversely, suppose that there is a shortest $s$-$t$ path $P$ of weight $2k$ satisfying $|N(V'(P))|\le rk-\binom{k}{2}+n-k$ in $G'$.
    If $P$ contains any vertex in $V_E$, the weight of $P$ is at least $2k+2>2k$. Thus, $P$ does not contain any vertex in $V_E$. Hence, by the construction of $G'$, $P$ must pass through $w_1,\cdots, w_{k-1}$ and exactly one vertex in each $V_1,\cdots, V_k$. We also observe that $V\setminus V'(P)\subseteq N(V'(P))$ holds. Consequently, $|N(V'(P))\setminus V_E|=n-k$ and $P$ has at most $rk-\binom{k}{2}$ neighbors in $V_E$.
    
    Let $S=V'(P)\cap V$. Since every vertex in $S$ has degree $r$ in $G$, the total number of incidences between vertices of $S$ and vertices of $V_E$ is $rk$. A vertex in $V_E$ is counted twice in this total if and only if it corresponds to an edge whose both endpoints belong to $S$, that is, to an edge of $G[S]$.
    Since every vertex of $V_E$ is adjacent to at most two vertices of $V$ in $G'$, no vertex of $V_E$ can be counted more than twice. Therefore, the number of neighbors of $P$ in $V_E$ is exactly $rk-|E(G[S])|$.
    Thus, we have $rk-|E(G[S])|\le rk-\binom{k}{2}$, which implies that $|E(G[S])|\ge \binom{k}{2}$. Since $|S|=k$ and $S$ contains exactly one vertex from each independent set $V_1,\dots,V_k$, the graph $G[S]$ can have at most $\binom{k}{2}$ edges.
    Therefore, we get $|E(G[S])|=\binom{k}{2}$.
    This means that $S(=V'(P)\cap V)$ forms a clique of size $k$, which completes the proof.
\end{proof}

Finally, we present an XP algorithm for \textsc{Shortest Secluded Path} that runs in $n^{O(d)}$ time where $d$ is the shortest path distance between $s$ and $t$ in the input graph.

\begin{theorem}\label{thm:weighted:XP}
     \textsc{Shortest Secluded Path} on edge-weighted graphs can be solved in $n^{O(d)}$ time,  where $d$ is the shortest path distance between $s$ and $t$ in the input graph.
\end{theorem}

\begin{proof}
Given an input graph $G=(V,E)$, we first compute the shortest path distance $d=\dist(s,t)$ between $s$ and $t$ by Dijkstra's algorithm~\cite{nm/Dijkstra59}.
Since the edge weights are nonnegative integers, any shortest $s$-$t$ path contains at most $d$ edges.
Thus, all shortest $s$-$t$ paths can be enumerated in $n^{O(d)}$ time.
Since the neighborhood of a shortest $s$-$t$ path can be computed in polynomial time, \textsc{Shortest Secluded Path} can be solved in $n^{O(d)}$ time.
\end{proof}

\section{Conclusion}
In this paper, we investigated the structural parameterizations of \textsc{Short Secluded Path} and the computational complexity of \textsc{Shortest Secluded Path}.

For the structural parameterizations of \textsc{Short Secluded Path}, we presented an XP algorithm parameterized by cliquewidth and FPT algorithms parameterized by neighborhood diversity and twin cover number, respectively.
Regarding \textsc{Shortest Secluded Path}, we first proposed a polynomial-time algorithm for unweighted graphs. For edge-weighted graphs, however, we proved that the problem is W[1]-hard but is in XP when parameterized by the shortest path distance between $s$ and $t$.

A natural future direction is to investigate the parameterized complexity of \textsc{Short Secluded Path} with respect to other structural parameters, such as the cluster vertex deletion number and modular-width. Furthermore, improving the running times of our algorithms is an interesting challenge. In particular, given an $r$-expression tree, our XP algorithm parameterized by cliquewidth runs in $n^{O(r^2)}$ time. It is worth considering whether this running time can be improved to $n^{O(r)}$ as in the case of \textsc{Hamiltonian Cycle}~\cite{algorithmica/BergougnouxKK20}. It would also be interesting to design a single-exponential FPT algorithm parameterized by the twin cover number, since our current algorithm runs in $2^{O(\tc(G)^2)}n^{O(1)}$ time.

\bibliographystyle{plain}
\bibliography{bibliography.bib}

\end{document}